%% file: root.tex
\documentclass[journal, final, letterpaper]{IEEEtran}

\usepackage{amsmath}
\usepackage{amssymb, amsfonts,amsthm}
\usepackage{booktabs} 
\usepackage{mathtools}
\usepackage{commath}
\usepackage{microtype}
\usepackage{multirow}
\usepackage{mathdots}
\usepackage{mathtools}
\usepackage{bm}

\usepackage{eucal}

\usepackage{tikz}
\usepackage{pgfplots}
\pgfplotsset{every axis/.append style={line width=1pt}}
\pgfplotsset{every tick label/.append style={font=\tiny}}
\pgfplotsset{compat=newest}
\usetikzlibrary{shapes, fit, decorations.markings, matrix, plotmarks, positioning, fit, spy, patterns, shadows, calc, backgrounds, arrows.meta}

\usetikzlibrary{positioning}

\usepackage[americanvoltages,siunitx]{circuitikz}

\usepackage[font=footnotesize,caption=false]{subfig}
\usepackage{graphicx}
\usepackage[abs]{overpic}

\usepackage{inputenc}
\usepackage[nolist]{acronym}
\usepackage{siunitx}
\usepackage{booktabs}

\usepackage{enumerate}
\newcounter{subeqn} %

\newtheorem{theorem}{Theorem}
\newtheorem{corollary}{Corollary}
\newtheorem{lemma}[theorem]{Lemma}
\newtheorem{proposition}{Proposition}
\newtheorem{assumption}{Assumption}
\newtheorem{definition}{Definition}

\newcommand{\ts}[1]{{\textnormal{#1}}}
\newcommand{\ie}{\emph{i.e.,}\ }

\newcommand{\Nset}{\mathbb{N}}
\newcommand{\Rset}{\mathbb{R}}

\newcommand{\mc}{\mathcal}
\newcommand{\mb}{\mathbf}
\newcommand{\mbb}{\mathbb}

\newif\ifproves
\provestrue

\allowdisplaybreaks

\begin{acronym}
	\acro{MPC}{Model Predictive Control}
	\acro{DMPC}{Distributed Model Predictive Control}
	\acro{LSS}{Large-Scale System}
	\acro{mRPI}{minimal Robust Positively Invariant}
	\acro{QP}{Quadratic Programming}
	\acro{RCI}{Robust Control Invariant}
	\acro{RPI}{Robust Positive Invariant}
	\acro{PnP}{Plug and Play}
	\acro{OCP}{Optimal Control Problem}
	\acro{MPC}{Model Predictive Control}
	\acro{PV}{Photo-voltaic}
	\acro{DG}{Distributed Generation}
	\acro{DS}{Distributed Storage}
	\acro{BIC}{Bounded Integral Controller}
	\acro{VCS}{Voltage controlled source}
	\acro{MG}{Micro-Grid}
	\acro{OPF}{Optimal Power Flow}
	\acro{SoC}{State of Charge}
	\acro{OCP}{Optimal Control Problem}
	\acro{CPL}{Constant Power Load}	
	\acro{MPPT}{Maximum Power Point Tracking}	
	\acro{WT}{Wind Turbine}
	\acro{PCC}{Point of Common Coupling}
	\acro{DER}{Distributed Energy Resource}
	\acro{MAS}{Multi Agent System}
        \acro{P2P}{peer-to-peer}
        \acro{ARIMA}{Auto-Regressive Integrated Moving Average}
\end{acronym}
\usepackage{cite}
\newlength{\hatchspread}
\newlength{\hatchthickness}
\newlength{\hatchshift}
\newcommand{\hatchcolor}{}
\tikzset{hatchspread/.code={\setlength{\hatchspread}{#1}},
         hatchthickness/.code={\setlength{\hatchthickness}{#1}},
         hatchshift/.code={\setlength{\hatchshift}{#1}},
         hatchcolor/.code={\renewcommand{\hatchcolor}{#1}}}
\tikzset{hatchspread=3pt,
         hatchthickness=0.4pt,
         hatchshift=0pt,
         hatchcolor=black}
\pgfdeclarepatternformonly[\hatchspread,\hatchthickness,\hatchshift,\hatchcolor]
   {nwl}
   {\pgfqpoint{\dimexpr-2\hatchthickness}{\dimexpr-2\hatchthickness}}
   {\pgfqpoint{\dimexpr\hatchspread+2\hatchthickness}{\dimexpr\hatchspread+2\hatchthickness}}
   {\pgfqpoint{\dimexpr\hatchspread}{\dimexpr\hatchspread}}
   {
    \pgfsetlinewidth{\hatchthickness}
    \pgfpathmoveto{\pgfqpoint{0pt}{\dimexpr\hatchspread+\hatchshift}}
    \pgfpathlineto{\pgfqpoint{\dimexpr\hatchspread+0.15pt+\hatchshift}{-0.15pt}}
    \ifdim \hatchshift > 0pt
      \pgfpathmoveto{\pgfqpoint{0pt}{\hatchshift}}
      \pgfpathlineto{\pgfqpoint{\dimexpr0.15pt+\hatchshift}{-0.15pt}}
    \fi
    \pgfsetstrokecolor{\hatchcolor}
    \pgfusepath{stroke}
   }
\pgfdeclarepatternformonly[\hatchspread,\hatchthickness,\hatchshift,\hatchcolor]
   {nwld}
   {\pgfqpoint{\dimexpr-2\hatchthickness}{\dimexpr-2\hatchthickness}}
   {\pgfqpoint{\dimexpr\hatchspread+2\hatchthickness}{\dimexpr\hatchspread+2\hatchthickness}}
   {\pgfqpoint{\dimexpr\hatchspread}{\dimexpr\hatchspread}}
   {
    \pgfsetlinewidth{\hatchthickness}
    \pgfpathmoveto{\pgfqpoint{0pt}{\dimexpr\hatchspread+\hatchshift}}
    \pgfpathlineto{\pgfqpoint{\dimexpr\hatchspread+0.15pt+\hatchshift}{-0.15pt}}
    \ifdim \hatchshift > 0pt
      \pgfpathmoveto{\pgfqpoint{0pt}{\hatchshift}}
      \pgfpathlineto{\pgfqpoint{\dimexpr0.15pt+\hatchshift}{-0.15pt}}
    \fi
    \pgfsetstrokecolor{\hatchcolor}
    \pgfsetdash{{1pt}{1pt}}{0pt}
    \pgfusepath{stroke}
   }
\pgfdeclarepatternformonly[\hatchspread,\hatchthickness,\hatchshift,\hatchcolor]
   {nel}
   {\pgfqpoint{\dimexpr-2\hatchthickness}{\dimexpr-2\hatchthickness}}
   {\pgfqpoint{\dimexpr\hatchspread+2\hatchthickness}{\dimexpr\hatchspread+2\hatchthickness}}
   {\pgfqpoint{\dimexpr\hatchspread}{\dimexpr\hatchspread}}
   {
    \pgfsetlinewidth{\hatchthickness}
    \pgfpathmoveto{\pgfqpoint{\dimexpr\hatchshift-0.15pt}{-0.15pt}}
    \pgfpathlineto{\pgfqpoint{\dimexpr\hatchspread+0.15pt}{\dimexpr\hatchspread-\hatchshift+0.15pt}}
    \ifdim \hatchshift > 0pt
      \pgfpathmoveto{\pgfqpoint{-0.15pt}{\dimexpr\hatchspread-\hatchshift-0.15pt}}
      \pgfpathlineto{\pgfqpoint{\dimexpr\hatchshift+0.15pt}{\dimexpr\hatchspread+0.15pt}}
    \fi
    \pgfsetstrokecolor{\hatchcolor}
    \pgfusepath{stroke}
   }
\pgfdeclarepatternformonly[\hatchspread,\hatchthickness,\hatchshift,\hatchcolor]
   {neld}
   {\pgfqpoint{\dimexpr-2\hatchthickness}{\dimexpr-2\hatchthickness}}
   {\pgfqpoint{\dimexpr\hatchspread+2\hatchthickness}{\dimexpr\hatchspread+2\hatchthickness}}
   {\pgfqpoint{\dimexpr\hatchspread}{\dimexpr\hatchspread}}
   {
    \pgfsetlinewidth{\hatchthickness}
    \pgfpathmoveto{\pgfqpoint{\dimexpr\hatchshift-0.15pt}{-0.15pt}}
    \pgfpathlineto{\pgfqpoint{\dimexpr\hatchspread+0.15pt}{\dimexpr\hatchspread-\hatchshift+0.15pt}}
    \ifdim \hatchshift > 0pt
      \pgfpathmoveto{\pgfqpoint{-0.15pt}{\dimexpr\hatchspread-\hatchshift-0.15pt}}
      \pgfpathlineto{\pgfqpoint{\dimexpr\hatchshift+0.15pt}{\dimexpr\hatchspread+0.15pt}}
    \fi
    \pgfsetstrokecolor{\hatchcolor}
    \pgfsetdash{{1pt}{1pt}}{0pt}
    \pgfusepath{stroke}
   }

\begin{document}
\title{Incorporating forecasting and peer-to-peer negotiation frameworks into a distributed model predictive control approach for meshed electric networks}

\author{Pablo R. Baldivieso Monasterios$^{1}$, Nandor Verba$^{2}$,  Euan A Morris$^{4}$, \\Thomas Morstyn$^{3}$\thanks{$^{3}$ Thomas Morstyn is with the School of Engineering, University of Edinburgh, Edinburgh, EH9 3JL, UK (e-mail: thomas.morstyn@ed.ac.uk).}, George. C. Konstantopoulos$^{1}$\thanks{$^{1}$ Pablo. R. Baldivieso Monasterios and George. C. Konstantopoulos are with the Department of Automatic Control \& Systems Engineering, University of Sheffield, Sheffield, UK {\tt\footnotesize \{p.baldivieso,g.konstantopoulos\}@sheffield.ac.uk}}, Elena Gaura$^{2}$\thanks{$^{2}$ Nandor Verba and Elena Gaura are with the  Faculty Research Centre for Computational science \& mathematical modelling, Coventry University, Coventry, UK {\tt\footnotesize \{ad2833,csx216\}@coventry.ac.uk}}, and Stephen McArthur$^{4}$\thanks{$^{4}$ Euan A. Morris and Stephen McArthur are with the Department of Electronic and Electrical Engineering, University of Strathclyde, Glasgow, UK {\tt\footnotesize \{euan.a.morris, s.mcarthur\}@strath.ac.uk}}
}
\maketitle

\begin{abstract}
The continuous integration of renewable energy sources into power networks is causing a paradigm shift in energy generation and distribution with regards to trading and control; the intermittent nature of renewable sources affects pricing of energy sold or purchased; the networks are subject to operational constraints, voltage limits at each node, rated capacities for the power electronic devices, current bounds for distribution lines. These economic and technical constraints coupled with intermittent renewable injection may pose a threat to system stability and performance. We propose a novel holistic approach to energy trading composed of a distributed predictive control framework to handle physical interactions, \ie voltage constraints and power dispatch, together with a negotiation framework to determine pricing policies for energy transactions. We study the effect of forecasting generation and consumption on the overall network's performance and market behaviours. We provide a rigorous convergence analysis for both the negotiation framework and the distributed control. Lastly, we assess the impact of forecasting in the proposed system with the aid of testing scenarios.
\end{abstract}


\section{Introduction}
\label{sec:introduction}

\IEEEPARstart{T}{he} current landscape of the electricity market is charaterised by the ever growing presence of renewable energy sources and a push for a deregulation of electricity markets. Both of these trends have a similar requirement: a decentralisation of operations, energy generation, control, and billing, which would grant more power to participants in the network \cite{Spence2018}. Decentralisation of control and pricing policies within power networks enhances reliability and flexibility, \ie the systems involved are more responsive to local changes, and modifications do not require global redesigns. These allow participants to engage with each other and perform energy transactions across the system. One of the challenges associated with implementing such decentralisation lie in combining pricing schemes with the control layer \cite{Farhangi2010}.

  A network encompassing renewable energy sources such as \ac{PV} panels, \ac{WT}, and batteries, takes centre stage in the modern energy generation paradigm which has as a defining feature its intermittent generation patterns. The uncertainty introduced by a \ac{DER} affects not only how each device is regulated but how energy is priced across the network. Furthermore several elements participating in a network act as prosumers, \ie entities capable of both generation and consumption of electricity. Traditionally, generation uncertainty could be handled using demand-side management techniques coupled with game theoretic concepts; for example in \cite{Mohsenian-Rad2010}, the authors propose scheduling of household appliances as strategies employed by users to drive electricity prices. Similarly, \cite{Wang2012} analyses the electricity price behaviour in response to the volatility of energy generation and transmission and concludes that the efficiency of market equilibria and average prices coincide with average marginal costs. As an alternative response to the volatility of renewable generation, \ac{P2P} trading platforms have emerged as a viable option, as seen in the Brooklyn \ac{MG} \cite{Mengelkamp2018}, to enable network members to transact energy surpluses or deficits granting flexibility to the grid. These trading platforms can be designed using techniques of game theory, auction theory, blockchain, and constrained optimisation \cite{Tushar2020}; in this paper, however, we focus on the game theoretic aspects of this problem. A concept from game theory heavily utilised in \ac{P2P} networks is that of the Stackelberg equilibrium. In \cite{Tushar2020a}, the authors propose and study a cooperative game theoretic framework where the grid acts as a leader and renewable sources as followers. A multiple leader, multiple follower structure was adopted in \cite{Paudel2019a} where authors show convergence to an equilibrium via simulations. Similar approaches are adopted by \cite{Zhang2019,Wang2019b,Anees2019} where the common thread among them is the existence of an stable Stackelberg equilibrium. The study of \ac{P2P} markets through the lens of game theory offers a rich framework that is capable of handling uncertainties arising from generation and consumption patterns; however, the relationship between these techniques and their physical implementations or relation to the control layers is not clear and requires a rigorous analysis.

  The relation between physical and market layers of a \ac{P2P} trading platform sheds light on the overall network behaviour, \ie how power can be transferred from node to node and the coordination needed to achieve an energy transaction. One crucial aspect to the physical layer is the presence of hard operational constraints imposed by the underlying power system \cite{Guerrero2018}. For example in \cite{Azim2019}, the authors show that careless operation, \ie not taking into account existing constraints, in the market layer may lead to voltage violations which in turn may compromise the entire network. A similar study, see \cite{Nikolaidis2018}, analyses the power flow of a network subject to \ac{P2P} and proposes allocation mechanisms to counteract the effect of power losses at transmission level. An alternative way to tackle this problem is given in \cite{Almasalma2019} where the authors of aim to link both market and control layers by employing similar \ac{P2P} protocols in the power converter voltage regulation; their approach relies heavily on distributed optimisation techniques which highlights the need of distributed methods for regulation at lower levels. Receding horizon techniques have been used before in the context of \ac{P2P} trading schemes as seen in \cite{Langer2020,Morstyn2019} where the prediction capabilities and the inherent constraint handling offer attractive properties for these type of trading platforms. Therefore, understanding the relation between both the market layer and its physical counterpart is crucial to ensure seamless operation of the \ac{MG}. Furthermore, the techniques required to regulate low level devices such as power converters require distributed controllers to maintain the spirit of \ac{P2P} trading. 

  In the excellent review \cite{Tushar2020}, the authors lay future possible research directions, among which they include the need for a unified model (incorporating both physical and market layers), the effect of injection limits in market mechanisms, and charge identification. This paper aims to contribute to bridge this gap in the understanding, using a rigorous mathematical framework, of how physical and market layers interact. Inspired by this framework, in the present paper, we propose a \ac{P2P} trading platform comprising a distributed receding horizon controller to handle the physical layer, and a negotiation based market layer with nonlinear pricing policies. Both control and negotiation layers incorporate forecasts available in their formulation, much in the spirit of \cite{Stephens2015}, to obtain  more reliable pricing and robust control actions, hence reducing the effect of volatile generation patterns. The control layer is capable of handling local constraints, \ie voltage levels or safety operation bounds, and coupled constraints, \ie transmission line capacities. We provide a rigorous framework to guarantee the recursive feasibility properties of our proposed distributed controller. On the other hand, the interaction between both market layer and physical layer is given by the pricing policy; once all prosumers reach an agreement on what pricing policy to use, this function is then passed to the receding horizon controller which incorporates it into its optimisation cost. Similarly to the physical layer, we rigorously analyse the properties of the market layer. The resulting approach combines both layers by feeding the pricing policy obtained in the market layer to the optimisation cost used by the distributed controller in the physical layer.
The contributions of this paper are the following:
\begin{itemize}
\item A rigorous mathematical framework for a novel distributed predictive controller managing interactions between local loads, \ac{DER}s, and the physical network. The controller handles coupled voltage constraints; uses forecasts of generation and demand in addition to neighbouring voltage information to compute its control law. The optimisation cost used by each \ac{DER} contains the pricing policy obtain in the market layer.  Section~\ref{sec:OCP} describes the different components of the \ac{OCP}.
\item A novel agent based market negotiation framework where all network elements engage in bargaining to determine suitable pricing policies used for energy transactions. This policies are used in the Distributed MPC cost such that performance of each \ac{MG} element is optimal with respect to this cost. We show that the outcome of this negotiation process is an Stackelberg equilibrium where the utility grid acts as the leader setting bounds on prices while all \ac{DER}s agree upon suitable policies. Section~\ref{sec:market} states the problem and relevant definitions of this negotiation problem.
\item A holistic algorithm that is capable to integrate the predictive controller,  agent based negotiation framework, and forecasting mechanisms in terms of convergence to game theoretic equilibrium concepts and recursive feasibility of predictive controllers. We propose a testing framework based on scenarios to better understand and evaluate the impact of forecasting approaches, pricing mechanisms, system configurations (\ie modifying rated values for some devices), and network topology on the value that an energy market coupled with a distributed control solution can offer.
\end{itemize}

\textbf{\emph{Notation}}:  For a given graph $\mc{G} = (\mc{V},\mc{E})$ with nodes $\mc{V}$ and edges $\mc{E}\subseteq\mc{V}\times\mc{V}$, the node-edge matrix $\mc{B}\in\Rset^{|\mc{E}|\times|\mc{V}|}$ characterises the relation between nodes and edges which for edge $e = (i,j)\in\mc{E}$ involving nodes $i$ and $j$ can be defined as $[\mc
{B}]_{ei} = 1$ if node $i$ is the source of $e\in\mc{E}$, and $[\mc{B}]_{ej} = -1$ if node $j$ is its sink, and zero otherwise. The $2-$norm is denoted $|x| = \norm{x}_2$. A \emph{C-set} is a compact and convex set containing the origin; A \emph{PC-set} is a C-set with the origin in its nonempty interior. For a given set $\mc{A}\subset\Rset^n$, and linear transformations $B\in\Rset^{m\times n}$ and $C\in\Rset^{n\times p}$, the image of $\mc{A}$ by $B$ is  $B\mc{A} = \{Bx\colon x\in\mc{A}\}\subset\Rset^m$ and the preimage of $\mc{A}$ by $C$ is $C^{-1}\mc{A} = \{x\colon Cx\in\mc{A}\}\subset\Rset^p$.



\section{Physical layer}
\label{sec:problem_setup}

In this section, we introduce the different components of the physical layer. The first component is composition of each prosumer (renewable sources, storage, and local loads), we, then, introduce the physical interconnection properties, and lastly we state the control objective from the perspective of the physical layer. Consider an undirected and connected graph $\mc{G} = (\mc{V},\mc{E})$ defining an electric network. The set of nodes $\mc{V}$ can be partitioned into two disjoint sets $\mc{V}_I$ and $\mc{V}_0$ corresponding to prosumers and utility electric grid respectively.
%
\subsubsection{System Model}
\label{sec:system-model}

Each prosumer $i\in\mc{V}_I$ comprises \ac{DER} sources interfaced via power converters and local loads, see Figure~\ref{fig:ADEPT}. The \ac{DER} sources operate with a maximum point tracking rationale enabling them to extract the maximum possible energy from given environmental conditions. The discrete power dynamics for each $i\in\mc{V}_I$ and $h\in\mc{H}\coloneqq\{\ts{wind},\ts{PV}\}$ are: 
\begin{equation}
S_{h,i}^+ = f_{h,i}(S_{h,i},\Delta_{h,i},w_{h,i})  
\label{eq:power_injection}
\end{equation}
where $S_{h,i}=(P_{h,i},Q_{h,i})$ denotes active and reactive power, $^+$ denotes the successor state, $w_{h,i}$ is uncontrollable input power generated by renewable sources, and $\Delta_{h,i}$ is a control input. Each prosumer has appended to it a battery with dynamics:
\begin{subequations}
\begin{align}
R_{b,i}C_{b,i}  (\ts{SoC}_i^+ - \ts{SoC}_i)  = & -(V(\ts{SoC}_i) - V_{b,i}),\\
R_{b,i}V_{b,i}V(\ts{SoC}_i) =&V_{b,i}^2 + R_{b,i}g_{b,i}(S_{b,i}),
\end{align}
\label{eq:BAT_dyn}	
\end{subequations}
%
where the state $(\ts{SoC}_i,V_{b,i})$ contains the \ac{SoC} and battery DC voltage; $V(\ts{SoC})$ is the \ac{SoC} dependent battery output voltage; $R_{b,i}$ is the internal resistance; and $C_{b,i}$ is the battery capacity in $[\si{\ampere\hour}]$. The nonlinear function $g_{b,i}(\cdot,\cdot)$ determines the power electronic steady-state behavior in terms of desired active and reactive desired power  $S_{b,i}=(P_{b,i},Q_{b,i})$ which we consider as inputs. The power electronic components from all renewable sources exhibit faster dynamic behavior, therefore we can consider each node operating in a \emph{quasi} stationary operation \cite{Venkatasubramanian1995}. The overall state for each node can be summarised in $x_i = (\{S_{h,i}\}_{h\in\mc{H}},{SoC}_{i})$, with control inputs $u_i = (\{\Delta_{w,i}\}_h,S_{b,i})$. Each \ac{DER} is subject to exogenous inputs $w_{g,i} = \{w_{h,i}\}_{h\in\mc{H}}$. In addition, local loads connected to node $i$ draw an \emph{a-priori} unknown active and reactive power $S_{l,i}=(P_{l,i},Q_{l,i})\in\Rset^2$; however, for each $i\in\mc{V}_I$, the controller has access to preview information, \ie forecasts for loads and \ac{DER} which satisfy the following Assumption:
\begin{assumption}[Forecasting information available]\ \ 
    \begin{enumerate}
        \item The state $x_i(k)$ and exogenous input $w_i(k) = (\{w_{h,i}(k)\}_{h\in\mc{H}},S_{l,i}(k))$ are known exactly at time $k$; future external inputs are not known exactly but satisfy $w_i(k+n) \in \mbb{D}_i$ for $n\in\Nset$.
        \item At any time step $k$, a prediction, $\mb{d}_i = \{d_i(k)\}_{k\in\Nset_{0:N-1}}$, of $N$ future exogenous inputs\footnote{$\Nset_{a:b} = \{a,a+1,\ldots,b-1,b\}$ for $a,b\in\Nset$ and $a<b$.}, over a finite horizon of time, is available.
    \end{enumerate}
\label{assum:preview_information}
\end{assumption}
Note that we do not assume anything about the accuracy of the predictions, and in fact will allow these to vary over time (this implicitly implies that previous predictions were not accurate). The states and inputs of each node are restricted to satisfy constraints $x_i\in\mbb{X}_i$ and $u_i\in\mbb{U}_i$, for which the following assumption hold
\begin{assumption}[Constraints]
For each $i\in\mc{V}_I$, the sets $\mbb{X}_i$ and $\mbb{D}_i$ are \emph{C-sets}. The set $\mbb{U}_i$ is a \emph{PC-set}.
\label{assum:constraints}
\end{assumption}
The output of each node $y_i = (P_{o,i},Q_{o,i})$ is given by its power balance equations $  y_i  = \sum_{h\in\mc{H}}S_{h,i}  + S_{b,i} - S_{l,i}$ which following Assumption~\ref{assum:constraints} is bounded for all time $k\in\Nset$.
\begin{figure}[t!]
\centering
\input{figs/local_diagram.tikz}
\caption{Power sources comprising node $\Sigma_i$; local energy sources together with local loads are connected to \ac{PCC}. \vspace{-0.5cm}}
\label{fig:ADEPT}	
\end{figure}
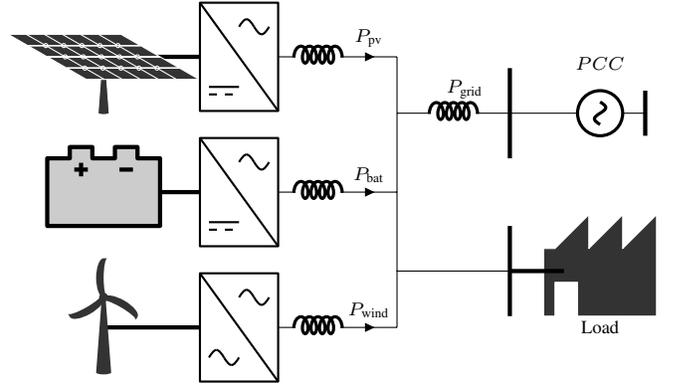
\subsubsection{Network model}
\label{sec:network-model}
The network topology is characterized by the set of edges $\mc{E}$ such that each $e\in\mc{E}$ defines the existence of a physical link between two nodes. The network topology allows us to define the set of neighbours of each $i\in\mc{V}$,
\begin{equation}
\mc{N}_i = \{j\in\mc{V}\colon (i,j)\in\mc{E}\}.
\label{eq:physical_neighbours}
\end{equation}
Similarly, we can define $\mc{E}_i\subseteq\mc{E}$ collecting all those edges emanating or terminating in $i\in\mc{V}$, \ie $\mc{E}_i = \{e\in\mc{E}\colon e = (i,j)\ts{ or }e = (j,i),~j\in\mc{V}\}$.
For each $i\in\mc{V}$, the current delivered by node $i$ is:
\begin{equation}
  i_i = \frac{v_i}{Z_{ii}} + \sum_{(i,j)\in\mc{E}_i} \frac{v_i - v_j}{Z_{ij}},
\label{eq:current_node}
\end{equation}
The admittance $Z_{ij}^{-1} = G_{ij} - jB_{ij}\in\mbb{C}$ corresponds to the line connecting the $i^\ts{th}$ and $j^\ts{th}$ nodes. The current drawn from each $i\in\mc{V}_I$ can be described in terms of $S_i = (P_i,Q_i)$, \ie active and reactive power, as\footnote{The matrix $\mbb{J}_2 = \begin{bmatrix}0 & -1\\1 & 0\end{bmatrix}$ is a complex structure on $\Rset^2$. Any complex number $a+ib$ can be written as a $2\times 2$ matrix or a vector in $\Rset^2$ as \[a+ib \iff \begin{bmatrix}a\\b\end{bmatrix} \iff \begin{bmatrix}a & -b\\b & a\end{bmatrix} = a \begin{bmatrix}1 & 0\\0 & 1\end{bmatrix} + b\begin{bmatrix}0 & -1\\1 & 0\end{bmatrix}\]}
\begin{equation}
i_i = h_i(v_i,S_i) = \frac{1}{|v_i|^2}\bigl( P_iv_i  - Q_i\mbb{J}_2v_i \bigr )
\label{eq:node_current_power}
\end{equation}
Furthermore, the node-edge incidence matrix can be partitioned into $\mc{B} = [\mc{B}_0~~\mc{B}_I]$ corresponding to the utility grid and prosumers. The current balance~\eqref{eq:current_node} for each $i\in\mc{V}_I$ is
\begin{equation}
  h_I(v_I,S_I) - (Y_I + \mc{L}_I)v_I - \mc{B}_I^\top Y_E\mc{B}_0 v_0 = 0,
  \label{eq:power_flow}
\end{equation}
where $\mc{L}_I = \mc{B}_I^\top Y_E\mc{B}_I$ with $Y_E = \ts{diag}\{Z_{ij}^{-1}\colon (i,j)\in\mc{E}\}$ the admittance of each line and $Y_I$ the shunt admittance of $i\in\mc{V}_I$. The vector $v_I = [v_i]_{i\in\mc{V}_I}\in\Rset^{2|\mc{V}_I|}$ collects all node voltages, similarly $S_I$ captures the power injected by prosumers, and $h_I(v_I,S_I) = [h_i(v_i,S_i))]_{i\in\mc{V}_I}\in\Rset^{|\mc{V}_I|}$. The grid voltage is given by $v_0\in\Rset^2$ and can be characterized by following:
\begin{equation}
-(Y_0 + \mc{L}_0)v_0 + Y_0 E_0 - \mc{B}_0^\top Y_E\mc{B}_Iv_I = 0.
\label{eq:grid_voltage}
\end{equation}
Similarly to the previous case, $Y_0$ is a local admittance, and $\mc{L}_0 = \mc{B}_0^\top Y_E\mc{B}_0$. The voltage $E_0$ is generated at the network connection point; the nature of this quantity varies according to the operation mode: stiff grid  $E_0 = (220\sqrt{2},0)$, a weak grid when its magnitude and angle are power dependent, or $E_0 = (0,0)$ in case of an islanded system.
\begin{figure}[!t]
\centering
\input{figs/global_diagram.tikz}
\caption{Physical network: each node $\Sigma_i$ is interfaced via inductive lines to a distribution network which may have a meshed topology. \vspace{-0.5cm}}
\end{figure}
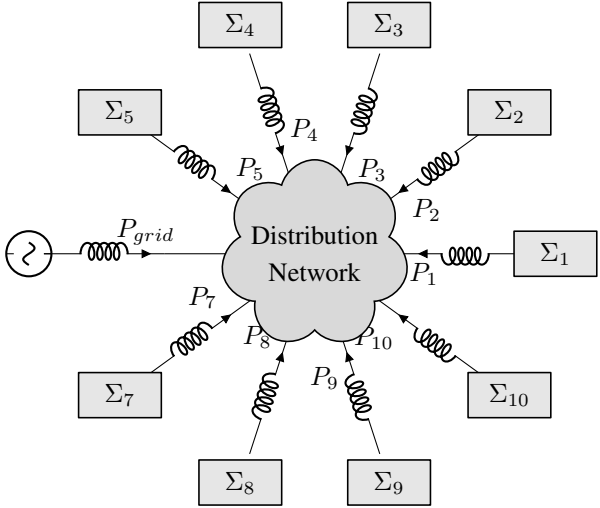
The network states are given by node voltages $v = (v_0,v_I)$ lying in a constraint set $\mbb{V} = \Rset^2\times\prod_{i\in\mc{V}_I}\mbb{V}_i$ where each set $\mbb{V}_i$ with $i\in\mc{V}_I$ satisfies
\begin{assumption}
The set $\mbb{V}_i\subset\Rset^2$ is a \emph{PC-set}.
\label{assum:network_constraints}
\end{assumption}
Similarly, if the currents flowing through the lines $i_E = Y_E\mc{B}v$ are also constrained to a PC-set $\mbb{I}_E = \prod_{e\in\mc{E}}\mbb{I}_e$, \ie bounds on each RMS current, these induce constraints on voltages by virtue of the algebraic relation $i_E = \mc{B}_0 v_0 + \mc{B}_I v_I$. Following~\eqref{eq:grid_voltage}, the prosumer voltage satisfies $v_I \in \tilde{A}^{-1}(-\tilde{B}E_0\oplus \mbb{I}_E)$, where $\tilde{A} = Y_E(\mc{B}_I - \mc{B}_0(Y_0 + \mc{L}_0)^{-1}\mc{B}_0^\top Y_E\mc{B}_I)$ and $\tilde{B} = Y_E\mc{B}_0(Y_0 + \mc{L}_0)^{-1}Y_0$. The overall constraint set is
\begin{equation}
  \mbb{V}_I =  \tilde{A}^{-1}(-\tilde{B}E_0\oplus \mbb{I}_E)\cap \prod_{i\in\mc{V}_I}\mbb{V}_i,
  \label{eq:voltage_constraints_network}
\end{equation}
and by virtue of Assumption~\ref{assum:constraints} is a PC-set.
\subsubsection{Physical control Objective}
\label{sec:control-objective}
The objective is twofold: find a suitable sequence of triplets  $(x,v,u)$ depending on external renewable injections $d(k)$ for $k\in\{0,1,\ldots\}$ that minimizes the infinite horizon criteria
\begin{equation}
J(x_0,v_0,d) = \sum_{k=0}^{\infty}\ell_k(x(k),u(k),d(k)) 
\label{eq:infinite_cost}
\end{equation}
where $\ell_k$ is a time-varying stage cost comprising generation costs. And second, to derive a suitable stage cost $\ell_k(\cdot,\cdot)$ that captures energy pricing mechanisms allowing for maximum profit at each node in terms of the exogenous inputs (load consumption and renewable injection) and available battery storage.
\vspace{-0.5cm}
\section{Peer-2-peer framework}
In this section, we introduce the proposed \ac{P2P} trading platform. We start by describing the physical controller in terms of the \ac{OCP}, we provide remarks on how the system operates and on its properties. We then proceed to state the proposed trading framework which is modelled as a game with one leader (utility grid) and multiple followers (prosumers). We first state the general approach, and then we propose a tractable reformulation.
\subsubsection{\ac{OCP} for energy system}
\label{sec:OCP}
To achieve our control objectives, consider the finite horizon criteria for each $i\in\mc{V}_I$ that employs exogenous predictions of Assumption~\ref{assum:preview_information}:
\begin{equation}
  J_{i}^{N}(\bar{z}_i,\mb{u}_i,\mb{d}_i) = \sum_{k=0}^{N-1}  \underbrace{\gamma(x_i,v_i,u_i) + \lambda_{k,i}(y_i) }_{\ell_{k,i}(z_i,u_i,d_i)} 
\label{eq:cost}
\end{equation}
with $z_i = (x_i,v_i)$, the stage cost $\ell_{i,k}(z_i,u_i,d_i)$ is composed by two terms capturing the costs of operating each node, in terms of batteries and voltages, and time-varying term $\lambda_{k,i}(\cdot)$ reflecting the cost of energy transactions. A standard Assumption on the stage cost for regularity purposes is
\begin{assumption}[Positive definite stage cost]
$\lambda_{i,k}\colon\mbb{X}_i\times\Rset^2\to\Rset$ and $\gamma_i\colon\mbb{U}_i\times\mbb{X}_i\to\Rset$ are, for each $i \in \mc{M}$ and $k\in\mbb{N}$, continuous positive definite functions.
\label{assump:pd}
\end{assumption}
The former penalises deviations from a reference voltage point and the cost of operating the battery which depend on the power value and state of charge. The latter function is determined from the negotiation framework, see Section~\ref{sec:market}, and  depends on the surplus or deficit power at each node. The performance criteria arguments are sequences of states $\mb{z}_i = \{z_i(0),\ldots z_i(N)\}$ and  controls $\mb{u}_i = \{u_i(0),\ldots u_i(N-1)\}$. Both of these sequences depend on exogenous inputs  $\mb{d}_i = \{d_i(0),\ldots, d_i(N-1)\}$; each $d_i(k)$ comprises predicted renewable injections $\{w_{h,i}(k)\}_{h\in\mc{H}}$, load consumption $S_{l,i}(k)$, and neighboring voltages $\{v_j(k)\}_{j\in\mc{N}_i}$. The resulting optimal control problem for each node for $\bar{z}_i = (\bar{x}_i,\bar{v}_i)$ and available predictions $\mb{d}_i$ is 
\vspace{-0.2cm}
\begin{equation}
  \mbb{P}_i(\bar{z}_i,\mb{d}_i)\colon\min \{J_{i}^{N}(\bar{z}_i,\mb{u}_i,\mb{d}_i) \colon \mb{u}_i\in\mc{U}_{i}^{N}(\bar{z}_i,\mb{d}_i)\}.
\label{eq:ocp}
\end{equation}
The constraint set $\mc{U}_{i}^{N}(\bar{z}_i,\mb{d}_i)$ is defined by 
\begin{subequations}
  \begin{align}
    (x_i &(0), v_i(0))  =  \bar{z}_i,\label{eq:cons:node_eq}\\
    x_{h,i}^+ &= f_{i}(x_{i},u_{i},d_i),\label{eq:cons:node_dyn}\\
        x_i\in\mbb{X}_i,\quad & u_i\in\mbb{U}_i, \quad v_i\in\mbb{V}_I(d_i),\label{eq:cons:ineq_cons}\\
    h_i(v_i,y_i) & = \mc{L}_i v_i + \hat{\mc{L}}_i d_i.\label{eq:cons:algebraic_cons}
  \end{align}
  \label{eq:cons}
\end{subequations}
The prediction model \eqref{eq:cons:node_dyn} differs from its counterpart~\eqref{eq:power_injection}--\eqref{eq:BAT_dyn} in the nature of the exogenous inputs; the former employs sequences of forecasts while the later uses the ``true'' values. This optimization problem is subject to coupled constraints~\eqref{eq:cons:ineq_cons} and \eqref{eq:cons:algebraic_cons} with respect to network voltages. The set $\mbb{V}_I(d_i)$ represents a ``slice'' of $\mbb{V}_I$ corresponding to node $i$ for given fixed values of neighboring voltages; when coupled constraints are absent, \ie no bounds on line currents, the voltage constraint sets are independent from neighbouring information. The matrix $\hat{\mc{L}}_i$ maps $d_i$ to the current balance for node $i$, \ie forming the $i^\ts{th}$ row of~\eqref{eq:power_flow} representing the power flow for node $i$. The solution of $\mbb{P}_i(\bar{z}_i,\mb{d}_i)$ is a sequence of optimal control inputs $\mb{u}_i^0$. One feature of this formulation is the introduction of cooperation between the \ac{MG} nodes; by sharing voltage information, each node $i$ is implicitly aware of power fluctuations from its physical neighbours. Suppose at time $k$, each node $i$ measures $z_i = ({x}_i,{v}_i)$, exchanges voltage prediction sequences, obtains forecasts for renewable injections and load demands such that $\mb{d}_i(k)$ is available; then applies the first element of the optimal solution $\mb{u}_i^0$ of~\eqref{eq:ocp} to the system. At the next sampling time, we discard the existing sequence, measure the plant, obtain new forecasts, then solve~\eqref{eq:ocp} with the updated information. This process is repeated \emph{ad infinitum}.  We discuss the recursive feasibility properties of the above \ac{OCP} in Section~\ref{sec:rec_feas}.
%
%
\subsubsection{Market negotiation}
\label{sec:market}
%


We propose an agent based approach to handle the market layer involving negotiations to choose adequate pricing policies. The market layer of our approach aims to handle power deficits and surpluses at each node $i\in\mc{V}_I$. The time-varying component of the cost~\eqref{eq:cost}, $\lambda_{i,k}(\cdot)$, weighs this output power and is the tool used to interface both the energy trading scheme and lower control levels. In our approach, we propose the use of a nonlinear pricing policy for exporting and importing power as opposed to the traditional linear pricing used in the literature \cite{Tushar2014}. This results in a negotiation framework where the participating nodes decide upon a policy which provides a mechanism to take predicted battery storage levels into account.

This negotiation framework can be interpreted as a game with a leader when the network is connected \cite{Lee2015}. The set of players is given by $\mc{V}_I\cup\mc{V}_0$; the set of actions for each $i\in\mc{V}_I$ is the set of functions $\mc{A}_i = \{\lambda_i \in\mc{L}_2[\mbb{Y}_i,\Rset]\colon \lambda_i (y_i) \geq 0\}$ with $\mbb{Y}_i$ the output constraint set which by Assumption~\ref{assum:preview_information} and \ref{assum:constraints} is compact. The choice of a non-negative function as the decision variable is directly linked with Assumption~\ref{assump:pd} which imposes regularity on the physical layer optimisation problem (existence of a solution to the \ac{OCP}). The action set for $\mc{V}_0$ is given by $\mc{A}_0 = \{(\lambda_{0,s},\lambda_{0,b})\in\Rset^{2}\colon \eta_s\leq \lambda_{0,s}<b_{0,b}\leq\eta_b\}$ representing the cost of purchasing or selling power to the grid. These prices are upper and lower bounded by $\eta_s>0$ and $\eta_b>0$ respectively. The total average network revenue is\footnote{The function $\sigma(x) = \frac{xe^{\alpha x}}{e^{\alpha x}+1}$ for a fixed $\alpha > 0$ is a smooth approximation of $\max(0,x)$. As $\alpha\to\infty$, $\sigma(\cdot)\to \max(0,\cdot)$.}
\begin{equation}
  R = \sum_{i\in\mc{V}_I}\frac{1}{\mu(\mbb{Y}_i)}\int_{Y_i}\lambda_i\circ\sigma d\mu
\label{eq:revenue_network}
\end{equation}
where each pricing policy is averaged over a set $Y_i= (y_i +S_{b,i}^\ts{max}[-\ts{SoC}_i,1-\ts{SoC}_i])\cap\mbb{Y}_i$ reflecting the available power with respect to current storage levels $\ts{SoC}$ and $\mu$ is the Lebesgue measure for $\mbb{Y}_i$. On the other hand, the power purchased by each $i\in\mc{V}_I$ acting as a buyer at a given time is
\begin{equation}
  \label{eq:power_purchased}
  \xi_i(\lambda_i,\lambda_{-i},y) = \frac{\lambda_i(-\sigma(-y_i))}{\sum\limits_{j\in\mc{N}_i^\ts{com}\cup\{i\}}\lambda_j(-\sigma(-y_j))}\sum\limits_{j\in\mc{N}_i^\ts{com}\cup\{i\}}\sigma(y_j) 
\end{equation}
which depends on neighbouring output power $y = (y_1,\ldots,y_{|\mc{V}_I|})$; it is worth noting that the only information needed to compute~\eqref{eq:power_purchased} is that of the trading neighbours. 
Each node behaves as a prosumer and it has attached to it a utility functional $r_i\colon\mc{A}_i\times\mc{A}_{-i}\to\Rset$, defined as
\begin{equation}
  \label{eq:utility_functional}
  \begin{split}
    r_i(\lambda_i,\lambda_{-i})  =  \sum_{k=0}^{N-1}  &\biggl(\lambda_{0,b}\xi_i(\lambda_i,\lambda_{-i},y(k))\\
    & - \lambda_i(\xi_i(\lambda_i,\lambda_{-i},y(k))) \\
    & + \log \bigl (1+\frac{\sigma(y_i(k))}{S_{b,i}^\ts{max}}\bigr) \\
    & + \gamma_iR(k)\frac{\sigma(y_i(k))}{1+\sum_{i\in\mc{V}_i}\sigma(y_i(k))}\biggr).
  \end{split}
\end{equation}
This utility functional measures revenue and satisfaction of each node with its current pricing policy, see \cite{Liu2017a} and \cite{Tushar2014}. There are two prominent parts: the first two terms represent the cost of purchasing with respect to the price set by the utility grid $\lambda_{0,b}>0$ which the agent seeks to minimise; the remaining terms correspond to the advantages of selling surplus power, the effect of available storage, and trade-off between increasing prices and loss of revenue. The choice of a logarithmic function represents a law of diminishing returns. The function $\sigma(\cdot)$ is a key component of this cost, depending on the output power sign, two terms will vanish implying that each node is either maximising profit or minimising costs but never both. This formulation avoids unwanted saddle or conservative behaviour when optimising. This approach induces a time-varying partition of $\mc{V}_I$ into two disjoint sets of sellers $\mc{V}_{I,S}$ when $y_i > 0$, and buyers $\mc{V}_{I,B}$ for $y_i < 0$.

The set $\mc{A}_{-i} \triangleq \prod_{j\in\mc{N}^\ts{com}_i}\mc{A}_j$  collects the actions of the neighbours in the communication network which are used to compute the total revenue $R(k)$ known by the utility grid, characterises the negotiation framework communication properties, and is defined as an unweighted graph characterised by a set of edges $\mc{E}^\ts{com}\subset\mc{V}\times\mc{V}$ which generates a set of neighbours $\mc{N}_i^\ts{com}\subset\mc{V}$ similar to \eqref{eq:physical_neighbours}. The particularity of this network is that for all $i\in\mc{V}_I$, the utility grid, if present, satisfies $0\in\mc{N}_i$ implying each negotiating prosumer can communicate with the utility grid; this condition is necessary since the utility grid agent determines the upper and lower bounds on electricity prices.
%
The utility grid solves the following optimisation problem with equilibrium constraints:
\begin{equation}
\mbb{P}_0^\ts{com}\colon \max\limits_{\lambda_0\in\mc{A}_0}\{r_0(\lambda_0,\lambda_{-0})\colon \lambda_{-0}\in\prod_{j\in\mc{N}_0}\mc{R}_j(\lambda_{-j})\}
\label{eq:leader_opt_problem}
\end{equation}
where the best reply map associated with $\lambda_{-j}$ is  $\mc{R}_j(\lambda_{-j}) = \{\lambda_j\in\mc{A}_j\colon r_j(\lambda_j,\lambda_{-j})\leq r_j(\tilde{\lambda}_j,\lambda_{-j}),~\forall\tilde{\lambda}_j\in\mc{A}_j\}$. The associated utility functional is
\[
  \begin{split}
    r_0(\lambda_0,\lambda_{-0}) = & \lambda_{0,s}\sum_{i\in\mc{V}_I}\sigma(y_i) + \lambda_{0,b}\sum_{i\in\mc{V}_I}\sigma(-y_i) \\
&     - \sum_{i\in\mc{V}_I}\lambda_i(\sigma(y_i))
  \end{split}\]
In this way, the utility grid agent sets the price according to the best response of the network members given by $\mc{V}_I$. Similarly, for each $i\in\mc{V}_I$ the corresponding optimisation problem is
\begin{equation}  \mbb{P}_i^\ts{com}(\lambda_{-j})\colon\max\limits_{\lambda_i\in\mc{A}_i}r_i(\lambda_i,\lambda_{-i})
\label{eq:follower_problem}
\end{equation}
An important feature of the presented approach and its relation to \ac{P2P} platforms lies in the way the revenue of energy transactions is obtained. After the agents are split into groups of buyers and sellers, the total revenue gained from a transaction is shared among the agents selling power. A method to establish direct contracts or bilateral negotiation schemes is the subject of ongoing research. 
\subsubsection*{Tractable reformulation}
The negotiation framework as stated in the previous section may be prohibitively difficult to solve. The action spaces for each $i\in\mc{V}_I$ are infinite dimensional spaces, and to exacerbate the problem, the optimisation problem~\eqref{eq:leader_opt_problem} has equilibrium constraints. We propose a method to simplify this problem to make it computationally tractable for both utility and network elements. Our first step towards this goal is to invoke the following:
\begin{assumption}
  Each node updates its pricing policy with a period $k_n\in\Nset$; the utility grid updates prices every $Hk_n\in\Nset$ for some $H\gg 0$.
  \label{assum:time_scale_negotiation}
\end{assumption}
This Assumption may refer to the common practice of setting a day ahead price based on existing consumption and is common in leader-follower games as mentioned in \cite{Fiez2020}. Furthermore, Assumption~\ref{assum:time_scale_negotiation} ensures that the utility grid is able to react only to the best replies from each network element. 

Another potential bottleneck, from an implementation point of view, is that of the action space infinite dimensionality $\mc{A}_i\subseteq\mc{L}_2[\mbb{Y}_i,\Rset]$. To overcome this hurdle, we propose to reduce the action space to a class of parameterised functions. The starting point is the traditional linear pricing policies, \ie $-\lambda_{0,b}y_i$ and $\lambda_{0,s}y_i$ for purchase and sale respectively. From Assumption~\ref{assum:constraints}, the output power is constrained to a bounded set. We seek a piece-wise smooth such that $\lambda_i(-y_i^\ts{max}) = \lambda_{0,b}y_i^\ts{max}$, $\lambda_i(0) = 0$, and $\lambda_i(y_i^\ts{max}) = \lambda_{0,b}y_i^\ts{max}$. Clearly, it is always possible to find a quadratic function fitting the positive part, and another fitting the negative one. The parameter we introduce is the deviation from a linear pricing: a way to measure this is to consider the area in between curves such that $b_i$ is \[ \int_0^{y_i^\ts{max}}\lambda_{0,s}y_idy - \int_0^{y_i^\ts{max}}\lambda_{i}dy =  \int^0_{-y_i^\ts{max}}-\lambda_{0,b}y_idy-\int^0_{-y_i^\ts{max}}\lambda_{i}dy.\]
Solving the above conditions yield the desired parameterised piece-wise smooth convex parameterisation. The missing ingredient is to satisfy Assumption~\ref{assump:pd}. This can be done by suitably constraining the available values for $b_i \in [0,b_i^\ts{max}]$ such that $\lambda_i(y_i,b_i^\ts{max}) \geq 0$ for all $y_i\in\mbb{Y}_i$. This allows us to assign a mapping between real positive numbers and $\mc{A}_i$ such that $b_i\mapsto\lambda_i(\cdot)$. Injectivity of this map follows naturally from construction; surjectivity, however, is not ensured since the image of the real numbers is not $\mc{A}_i$ but only a strict subset. This discrepancy is because of the inequality condition used to parameterise desired positive definite functions. In this way and owing to the continuity of $b\mapsto\lambda_i(\cdot)$, the problem of optimising over function spaces is reduced to optimising over $\Rset^{|\mc{V}_I|}$. Following Assumption~\ref{assum:time_scale_negotiation}, the game can be played as: 
\begin{enumerate}
  \item At time $k=0,~n=0$, price bounds $\lambda_0 = (\eta_s,\eta_b)$.
  \item\label{it:nash} if $\mod(k,k_n) = 0$, \textbf{do} for all $i\in\mc{V}_I$ solve \eqref{eq:follower_problem} using a best reply updating\\ 
    \textbf{until} A Nash equilibrium $\{\lambda_i^0(\cdot)\}_{i\in\mc{V}_I}$ is found. 
  \item if $mod(k,Hk_n)=0$, \textbf{do} Utility grid solves~\eqref{eq:leader_opt_problem} and updates $\lambda_0$.
  \item $k\to k+1$, go to \ref{it:nash}.
  \end{enumerate}
  The utility grid updates its decision variables in response to optimal pricing profiles obtained by network agents. The negotiation between network and utility grid agents occurs on top of the control layer described in Section~\ref{sec:OCP}. \vspace{-0.3cm}
%
\section{Convergence analysis}
\label{sec:stab-conv-analys}
In this section, we analyse the theoretical properties of the proposed pricing approach. We divide our analysis into two main parts: market negotiation convergence, and recursive feasibility. We finish this Section with remarks on how both market and control layers interact.
\subsection{Negotiation Convergence}
\label{sec:market_neg}
In this section, we study the convergence properties of  $\bigl(\mc{V}_I,\{\mc{A}_i\}_{i\in\mc{V}_I},\{r_i\}_{i\in\mc{V}_I}\bigr)$ as a Stackelberg game. The outcome of this game will define the pricing policy used by each element $i\in\mc{V}_I$. We begin this part of the analysis by defining the equilibrium concepts that will be used: 
\begin{definition}[Nash equilibrium]
  An action profile $\lambda^0 = (\lambda^0_{1},\dots,\lambda^0_{|\mc{V}_I|})$ is said to be a \emph{Nash equilibrium} of the game
  $\bigl(\mc{V}_I,\{\mc{A}_i\}_{i\in\mc{V}_I},\{r_i\}_{i\in\mc{V}_I}\bigr)$ if, for all
  $i \in \mc{V}_I$,
\begin{equation}
  r_i(\lambda^0_{i},\lambda^0_{-i}) = \min_{\lambda_{i} \in \mc{A}_i} r_i(\lambda_{i},\lambda^0_{-i}). 
\end{equation}
\label{def:nash}
 \vspace{-0.6cm}
\end{definition}
The set of Nash equilibrium points for $\mc{V}_I$ parameterised by $\lambda_0\in\mc{A}_0$ are $\mbb{NE}(\lambda_0)\subset\prod_{i\in\mc{V}_I}\mc{A}_i$, this set represents the best network response to the prices set by the utility grid. This concept leads us to the other equilibrium concept we leverage on: 
\begin{definition}[Stackelberg equilibrium]
  An action profile $\lambda^*=(\lambda_{0}^*,\ldots,\lambda_{M}^*)$ is said to be an \emph{Stackelberg equilibrium} of the $1$ leader, $M-$follower game $\bigl(\mc{V},\{\mc{A}_i\}_{i\in\mc{V}},\{r_i\}_{i\in\mc{V}}\bigr)$ if for all $i\in\mc{V}$
\begin{equation} \sup\limits_{\lambda_{-0}\in\mbb{NE}(\lambda_0^0)}r_i(\lambda_0^0,\lambda_{-0})\leq\sup\limits_{\lambda_{-0}\in\mbb{NE}(\lambda_0)}r_i(\lambda_0,\lambda_{-0}) 
\label{eq:st_eq_leader}
\end{equation}
\vspace{-0.3cm}
\label{def:stackelberg}  
\end{definition}
The Stackelberg equilibrium complements that of the Nash equilibrium and essentially leads to an optimal response from the utility grid side in response to the best possible actions from the network side. Following Assumption~\ref{assum:time_scale_negotiation}, it is possible to solve these two problems independently with the caveat that it leads to a problem with equilibrium constraints, see\cite{Ehrenmann2004} for an in-depth study of this type of problems. Given the utility functionals $r_i$ naturally partition the set of nodes, it is possible without loss of generality to choose the buyer nodes to play first. Owing to the parameterisation of the pricing policy, $b_i\mapsto\lambda_i(\cdot)$, is convex by construction. Our first result considers the case when $y_i<0$, \ie node $i$ is purchasing power.
\begin{lemma}
  Suppose $y_i < 0$, then the purchased power in \eqref{eq:power_purchased} is a concave function with respect $b_i\in [0,b_i^\ts{max}]$. 
\label{lem:buyer_convex}  
\end{lemma}
\begin{proof}
  By construction, $\forall i\in\mc{V}_I$ $\lambda_i(y_i,b_i)\geq 0$ and since the parameterisation is a linear problem, the policy is linear with respect to the parameter. Indeed, the conditions defined in Section~\ref{sec:market_neg} for the negative part of the policy $\lambda_i(y_i) = a_2y_i^2 + a_1y_i$ are \[\begin{split}
      a_2(y_i^\ts{max})^2 - a_1y_i^\ts{max} & = \lambda_{0,b}y_i^\ts{max}\\
\frac{1}{2}\lambda_{0,b}(y_i^\ts{max})^2 + \frac{1}{3}a_2(y_i^\ts{max})^3 + \frac{1}{2} a_1(y_i^\ts{max})^2  &= b_i      
  \end{split}\]
This yields a policy\[\lambda_i(y_i,b_i) = \frac{6b_i}{(y_i^\ts{max})^3}y_i^2 +\frac{6b_i-\lambda_{0,b}(y_i^\ts{max})^2}{(y_i^\ts{max})^2}y_i.\]
This implies that the policy is concave on $b_i$. To prove the purchased power concavity, it is enough to check its derivatives. A direct calculation shows that \[\frac{\partial^2\xi_i}{\partial b_i^2} = -2 \frac{\sum_{l}\sigma(y_l)\sum_{l\neq i}\lambda_l(y_l,b_l)\bigl(\frac{\partial\lambda_i}{\partial b_i}\bigr)^2}{(\sum_{j\in\mc{N}_i^\ts{com}\cup\{i\}}\lambda_j(y_j))^3}<0\]
is negative definite.
\end{proof}
\begin{lemma}
  Suppose $b_i \in [0,b_i^\ts{max}]$ for all $i\in\mc{V}_I$, then $r_i(\lambda_i,\lambda_{-i})$ is a concave function with respect to $b_i$ and $b_{-i}$.
\label{lem:concave_utility}  
\end{lemma}
\begin{proof}

  The proof proceeds by cases: $i)$ $y_i\geq 0$. In this case, the first two terms of $r_i(\cdot,\cdot)$ vanish. The total revenue is by construction a linear function of $b_i$ and $b_{-i}$ and the logarithmic term is concave with respect to its domain. The resulting function which is a sum of concave functions is therefore concave.

  $ii)$ $y_i < 0$. In this case only the first two terms of $r_i(\cdot,\cdot)$ contribute, the later vanish. Following Lemma~\ref{lem:buyer_convex}, the purchased energy is a concave function of $b_{i}$. The argument follows, \emph{mutatis mutandis}, that of Lemma~\ref{lem:buyer_convex} to obtain the negativity of the second derivative within the range $[0,b_i^\ts{max}]$.
\end{proof}
These two Lemmas lead to our first result:
\begin{theorem}[Network Nash equilibrium]
  The game $(\mc{V}_I,\{\mc{A}_i\}_{i\in\mc{V}_I},\{r_i\}_{i\in\mc{V}_I})$ admits a \emph{Nash equilibrium}.
\label{thm:Nash_eq}
\end{theorem}
The proof of this theorem follows from \cite[Theorem 4.3]{Basar1998} which is a theorem that guarantees the existence of a Nash equilibrium point with pure strategies given the concavity of the cost functionals. Therefore Theorem~\ref{thm:Nash_eq} guarantees the existing of an optimal piece-wise smooth pricing policy that is dependent on the prices set by the utility grid. The move played by the grid satisfies $\lambda_0^0 = \arg\max r_0(\lambda_0,\lambda_{-0})$ which is by construction a linear problem over a compact set. A consequence of this fact is the existence of a Stackelberg equilibrium between network and utility grid. In summary, the solution of the bi-level optimisation problem \eqref{eq:leader_opt_problem} and \eqref{eq:follower_problem} follows an algorithmic procedure. First, initial utility $\lambda_0(0)$ prices are given; the followers, \ie the elements of $\mc{V}_I$, solve their optimisation problems to find a Nash equilibrium. Once this equilibrium is achieved, the utility grid responds to it by setting new prices for the new iteration $\lambda_0(1)$ and the process repeats until convergence. This approach to pricing guarantees that all $i\in\mc{V}_I$ price their energy in an optimal way according to their own available power, in terms of storage, and information from its neighbours. 
\subsection{Recursive feasibility}
\label{sec:rec_feas}
From Assumption~\ref{assum:preview_information} and \ref{assum:network_constraints}, the exogenous information available to the controller is contained within a PC-set $d_i(k)\in\mbb{D}_i$ for all $k\geq 0$. The set $\mc{D}_{i}^{N} = \prod_{k = 0}^{N-1}\mbb{D}_i$ contains all sequences of length $N$. The challenge arises as a consequence of the receding horizon implementation of the control action: if $\mbb{P}_i(z_i,\mb{d}_i)$ is feasible and yields a solution sequence $\mb{u}_i^0(z_i,\mb{d}_i)$, then its first element, which defines an implicit control law $\kappa_{i}^{N}(z_i,\mb{d}_i)) = u_i^0(0;z_i,\mb{d}_i)$, is applied to the system resulting in the state evolution $z^+ = (x_i^+,v_i^+)$ obtained as a solution of the difference algebraic model~\eqref{eq:power_injection},~\eqref{eq:BAT_dyn}, and~\eqref{eq:current_node}. At this sampling instant a new sequence of information is available to the controller, \ie $\mb{d}_i^+\in\mc{D}_{i}^{N}$, and the problem to be solved is now $\mbb{P}_i(z_i^+,\mb{d}_i^+)$. In this Section, we seek an answer to the question: What are the conditions necessary to ensure $\mbb{P}_i(z_i^+,\mb{d}_i^+)$ has a solution when $\mb{d}_i$ changes (perhaps arbitrarily) to $\mb{d}_i^+$?

We begin by defining \emph{recursive feasibility}, and some useful terminology in the analysis:
\begin{definition}[Recursive feasibility]
For each $i\in\mc{V}$, the \ac{OCP} $\mbb{P}_i(z_i,\mb{d}_i)$ is said to be \emph{recursively feasible} if $ \mc{U}_{i}^{N}(z_i,\mb{d}_i)\neq\emptyset$, then for a successor state $z_i^+ = (x_i^+,v_i^+)$, and $\mb{d}_i^+\in\mc{D}_{i}^{N}$, the constraint set  $\mc{U}_{i}^{N}(z_i^+,\mb{d}_i^+)\neq\emptyset$. 
\label{def:recursive_feas}
\end{definition}
For a disturbance sequence at time $k$, $\mb{d}_i = \{d_i(0),\ldots,d_i(N-1)\}$, its associated $\tilde{k}^\ts{th}$ tail for time $k+1$ is $\tilde{d}_{\tilde{k}}(\mb{d}_i) = \{d_i(1)\ldots,d_i(\tilde{k}-1),d_i(\tilde{k}-1),d_i(\tilde{k}),\ldots,d_i(N-1)\}$. It is clear that both $\mb{d}_i$ and $\tilde{d}(\mb{d}_i)$ belong to the set $\mbb{D}_i$; the notion of the tail allows us to quantify the change in information that a controller is subject to, for two sequences $\mb{d},\mb{e}\in\mc{D}_{i}^{N}$, the distance 
$\rho(\mb{d},\mb{e}) = |\tilde{d}(\mb{d}) - \mb{e}|$ is a metric on the sequence space $\mc{D}_{i}^{N}$.


A system is said to be locally controllable at a point $z_0\in Z_0$ if for every $\varepsilon >0$, $H\in\Nset$ and $\bar{z}$ such that $|z - z_0|\leq \varepsilon$, there exists a finite sequence of controls $\{u(0),\ldots,u(H-1)\}$ such that its solution satisfies $|z(k) - z_0|<\varepsilon$ for all $j\in\Nset_{0:H-1}$ with $z(0) = z_0$ and $z(H) = \bar{z}$. The set of feasible states is $\mc{Z}_{i}^{N}(\mb{d}_i) = \{ z_i\colon\mc{U}_{i}^{N}(z_i,\mb{d}_i)\neq\emptyset\}$ defines the region in the state space such that the \ac{OCP} is feasible. The first of our results is concerned when the exogenous information is unchanging, \ie the future information is taken to be the tail of the initial sequence.

\begin{theorem}
Let $\mb{d}_i\in\mc{D}_{i}^{N}$ and suppose each node is locally controllable with respect to $\mb{d}_i$. If $\mb{d}_i^+ = \tilde{d}_{\tilde{k}}(\mb{d}_i)$, then $(x_i,v_i) \in \mc{Z}_{i}^{N}(\mb{d}_i)$ implies $(x_i^+,v_i^+) \in \mc{Z}_{i}^{N}(\mb{d}_i^+)$.
\label{thm:feasibility_tail}
\end{theorem}
\begin{proof}
  Given $(x_i,v_i) \in \mc{Z}_{i}^{N}(\mb{d}_i)$, then an optimal sequence of control actions $\mb{u}_i^0(z_i,\mb{d}_i) = \{u_i^0(0;z_i,\mb{d}_i),u_i^0(1;z_i,\mb{d}_i),\ldots,u_i^0(N-1;z_i,\mb{d}_i)\}$ exists and generates a sequence of states and voltages $\mb{x}_i^0 = \{x_i^0(0),\ldots,x_i^0(N)\}$ and $\mb{v}_i^0 = \{v_i^0(0),\ldots,v_i^0(N-1)\}$ for a given information $\mb{d}_i\in\mc{D}_{i}^{N}$. Now, we construct a sequence of control actions $\tilde{\mb{u}}_i = \{\tilde{u}(0),\ldots,\tilde{u}(N-1)\}\in\mc{U}_i^{N}(z_i^+,\mb{d}_i^+)$. Using the definition of a disturbance tail, the first $\tilde{k}$ elements of $\mb{d}_i^+$ satisfy $\mb{d}_i^+(k) = \mb{d}_i(k+1)$ implying that $\tilde{\mb{u}}_i(k) = u_i^0(k+1;z_i,\mb{d}_i)$ for all $k\in\{0,\ldots,\tilde{k}-1\}$. The state evolution is governed by a set of difference algebraic equations $F_i(z_i,u_i,d_i,z_i^+) = 0$ composed of $\mc{C}^1$ dynamics~\eqref{eq:power_injection}--~\eqref{eq:BAT_dyn} and \eqref{eq:power_flow}. Using the implicit function theorem, it is possible to locally define a function $\xi_V(\cdot,\cdot,\cdot)$ such that $z_i^+ = \xi_V(z_i,u_i,d_i)$ and ensure the existence of neighbourhoods $V$ and $Z$ such that for $(z_i,u_i,d_i)\in V$, $(z_i,u_i,d_i,\xi_V(z_i,u_i,d_i))\in Z$ and $F_i(z_i,u_i,d_i,\xi_V(z_i,u_i,d_i))=0$. Consider the initial state to be $\tilde{z}_i(0)= (x_i^0(1),v_i^0(1))$, the subsequent states, following $\xi_V$, satisfy $\tilde{z}_i(k) = (x_i^0(k+1),v_i^0(k+1))$ for all $k\in\{0,\ldots,\tilde{k}\}$. The next element satisfies $\tilde{z}_i(\tilde{k}+1) = \xi_V(\tilde{z}_i(\tilde{k}),\tilde{u}_i,d_i(\tilde{k}))$ for some $\tilde{u}_i\in\mbb{U}_i$. Since each node is locally controllable with respect to $\mb{d}_i$, there exists a control action $\tilde{u}_i\in\mbb{U}_i$ such that $\xi_V(z_i,\tilde{u}_i,d_i)\in\mbb{X}_i\times\mbb{V}_I(d_i)$. Similarly, there exists a control law $\hat{u}\in\mbb{U}_i$ such that $\tilde{z}_i(\tilde{k}+2)= z_i^0(\tilde{k}+2) = \xi_V(z_i(\tilde{k}+1),\hat{u},d_i(\tilde{k}+1))$. The resulting sequence satisfies $
      \tilde{\mb{u}}_i = \{u_i^0(1),\ldots,u_i^0(\tilde{k}),\tilde{u}_i, 
       \hat{u}_i,u_i^0(\tilde{k}+2),\ldots,u_i^0(N-1)\} \in\mc{U}(z_i^+,\mb{d}_i^+)$.
    
\end{proof}
Once recursive feasibility of the tail is achieved, we allow the information to vary (perhaps arbitrarily). To study this case, we turn to the continuity properties of both value function and constraints which depend on exogenous inputs and initial states. The set $\Gamma_{i}^{N} = \{(z_i,\mb{d}_i)\colon z_i\in\mc{Z}_{i}^{N}(\mb{d}_i),~\mb{d}_i\in\mc{D}_{i}^{N}\} $ is the graph $\ts{gr }\mc{Z}_i^N$ of the set-valued map corresponding to the feasible set $\mc{Z}_i^N\colon\mc{D}_i^N\to 2^{\mbb{X}_i}$.

\begin{definition}[Upper semicontinuity for set-valued maps]
  A set $\Phi\colon U\to 2^X$ is upper semicontinous at $\xi_0\in U$ if for an open neighbourhood $V_U\subset U$ of $\xi_0$, then for all  $\xi\in\ V_U$  $\Phi(\xi)\subset V_X$ for an open neighbourhood $V_X\subset X$.
\label{def:usc}  
\end{definition}

\begin{proposition}
  Suppose Assumption~\ref{assum:preview_information}--~\ref{assump:pd} hold and the dynamics are continuous. Then the set valued map $\mc{Z}_I^N(\cdot)$ is upper semicontinuous.
\label{prop:Z_usc}  
\end{proposition}
\begin{proof}
  The constraints given by~\eqref{eq:cons}, by Assumptions~\ref{assum:constraints} and \ref{assum:network_constraints} together with  the continuity of each node's dynamics, have a structure $\mc{U}_i^N(\bar{z}_i,\mb{d}_i) = \{\mb{u}_i\colon G_i(\mb{u}_i,\bar{z}_i,\mb{d}_i)\in\mbb{K}\}$ for a fixed compact set $\mbb{K}$ and a continuous function~$G(\cdot,\cdot,\cdot)$. This implies that the graph of $\mc{U}_i^N(\cdot,\cdot)$ is a compact set; since the underlying space is finite dimensional, there is a compact neighbourhood $V_\Phi$ such that $\ts{gr }\mc{U}_i^N \subset V_\Phi$. If the set $\Gamma_i^N$ is closed and $\mbb{X}_i\times\mbb{V}_I(\mb{d}_i)$ is compact then by \cite[Lemma 4.3]{Bonnans2000a} our result follows. Consider a converging  sequence $\{\bar{z}_{i,k},\mb{d}_{i,k}\}\subset\Gamma_i^N$; to prove our result, we need to ensure that the limit point $(\bar{z}_i,\mb{d}_i)$ belongs to $\Gamma_i^N$. Using the definition of $\Gamma_i^N$, there exists $\mb{u}_{i,k}$ such that $G_i(\mb{u}_{i,k},\bar{z}_{i,k},\mb{d}_{i,k})\in\mbb{K}$ which by continuity of $G_i(\cdot,\cdot,\cdot)$ yields $\mb{u}_{i,k}\to\bar{\mb{u}}_i $ and  $G_i(\bar{z}_i,\bar{\mb{u}},\mb{d}_i)\in\mc{K}$. The closedness of $\Gamma_i^N$ follows.
\end{proof}
The feasible map $\mc{Z}_i^N(\cdot)$ is the domain of the constraint map $\mc{U}_i^N(\cdot)$. The proof of proposition~\ref{prop:Z_usc} shows the close relation between the set defining constraints and the feasible region; moreover, this results implicitly states that for small deviation in the exogenous inputs $\mb{d}_i$, the resulting optimisation problems associated to the feasible sets have a solution.

The following result builds on the previous ones:
\begin{proposition}[Value function continuity]
  Suppose Assumption~\ref{assum:preview_information}--~\ref{assump:pd} hold, $f_i(\cdot)$ and $h_i(\cdot)$ are continuous, and the set of optimisers of $\mbb{P}_i^N(\bar{z}_i,\mb{d}_i)$ is compact. Then the value function
  \begin{equation}
    \nu_i^{N,0}(\bar{z},\mb{d}_i) = \min\{J_i^N(\bar{x}_i,\mb{u}_i,\mb{d}_i)\colon \mb{u}_i^N\in\mc{U}_i^N(\bar{z}_i,\mb{d}_i)\}
    \label{eq:value_function}
  \end{equation}
is continuous.  
\label{prop:value_function_continuity}  
\end{proposition}
\begin{proof}
  The proof is an adaptation of  \cite[Proposition 4.4]{Bonnans2000a} to our setting. The continuity of the value function depends on the upper semicontiuity of the constraint set, which holds by Proposition~\ref{prop:Z_usc}, and the existence of neighbourhoods of the set of optimisers of~\eqref{eq:ocp}, which is guaranteed since these form a compact set by assumption. The compactness of implies that there exists a finite open covering of $\{(\bar{z}_i,\mb{d}_i)\}\times\mc{S}(\bar{z}_i,\mb{d}_i)\subset V_Z\times V_U$, where $\mc{S}(\bar{z}_i,\mb{d}_i) = \arg\min \{J_i^N(\bar{x}_i,\mb{u}_i,\mb{d}_i)\colon \mb{u}_i^N\in\mc{U}_i^N(\bar{z}_i,\mb{d}_i)\}$. In this neighbourhood, the ``almost'' optimal points that satisfy $J_i(\tilde{x}_i,\mb{u}_i,\tilde{\mb{d}}_i) \leq \nu_i^{N,0}(\bar{z}_i,\mb{d}_i) + \varepsilon$ for all $(\tilde{z},\tilde{\mb{d}}_i,\mb{u}_i)\in V_Z\times V_U$. Since the neighbourhood $V_U$ contains an optimal point which belongs to a closed set, the intersection $V_U\cap\mc{U}_i^N(\tilde{z}_i,\tilde{\mb{d}}_i) \neq \emptyset$ which yields: $\nu_i^{N,0}(\tilde{z}_i,\tilde{\mb{d}}_i) \leq \nu_i^{N,0}(\bar{z}_i,\mb{d}_i) + \varepsilon$.

On the other hand, $\nu_i^{N,0}(\bar{z}_i,\mb{d}_i) - \varepsilon \leq J_i(\bar{x}_i,\mb{u}_i,\mb{d}_i)$ holds for all $\mb{u}_i\in\mc{U}_i^N(\bar{z}_i,\mb{d}_i)$. From the proof of Proposition~\ref{prop:Z_usc}, the set $\mc{U}_i^N(\cdot,\cdot)$ is upper semicontinuous and there exists neighbourhoods $V_{U'}$ and $V_{Z'}$ such that for all $(\tilde{z}_i,\tilde{\mb{d}}_i)\in V_{Z'}$ and  $\mb{u}_i\in V_{U'}\cap\mc{U}_i^N(\tilde{z}_i,\tilde{\mb{d}}_i)$, $\nu_i^{N,0}(\bar{z}_i,\mb{d}_i) - \varepsilon \leq J_i(\tilde{x}_i,\tilde{\mb{u}_i},\tilde{\mb{d}}_i)$ holds. Since $\varepsilon >0$ is arbitrary, then $\nu_i^{N,0}(\bar{z}_i,\mb{d}_i) - \varepsilon \leq  \nu_i^{N,0}(\tilde{z}_i,\tilde{\mb{d}}_i)$. Continuity of $\nu_i^0(\cdot)$ follows.
\end{proof}

The continuity of the value function is a crucial property for our objective. Consider a subset $\Omega_{i,\beta}= \{(z_i,\mb{d}_i)\in\Gamma_i^N\colon \nu_i^{N,0} ((z_i,\mb{d}_i)) \leq \beta\}$ for $\beta >0$.

\begin{assumption}
The exogenous input sequence evolves as $\mb{d}_i^+ = \tilde{d}_{\tilde{k}}(\mb{d}_i) + \Delta\mb{d}_i$ where $\Delta\mb{d}_i = \mb{d}^+ - \tilde{d}_{\tilde{k}}(\mb{d}_i)\in\Delta\mc{D}_i^N$. The set $\Delta\mc{D}_i^N$ is chosen such that\footnote{For a set $\mc{A}\subset\mc{B}$, its diameter with respect to $\mc{B}$ is $\ts{diam}_\mc{B}\mc{A} = \max\{|x-y|\colon x,y\in\mc{B},x-y\in\mc{A}\}$.} $\lambda_i = \ts{diam}_{\mc{D}_i^N}\Delta\mc{D}_i^N$ satisfies $\lambda_i \leq \sigma_{\nu,i}^{-1}((id - \gamma_i)(\alpha_i))$ where $\gamma_i$ is a $\mc{K}-$function, and $\alpha_i >0$ such that $\Omega_{i,\alpha}\subset\Gamma_i^N$.
\label{assum:bounded_disturbance}  
\end{assumption}

The overall dynamics for both states and disturbance satisfy
\begin{subequations}
\begin{align}
  F_i(z_i,\kappa_i^N(z_i,\mb{d}_i),d_i,z_i^+) & = 0\label{eq:state}\\
  \mb{d}_i^+ \in \tilde{d}_{\tilde{k}}(\mb{d}_i) + \Delta\mc{D}_i^N\label{eq:dist}
\end{align}
\label{eq:overall_diff_inclusion}
\end{subequations}
We claim that a subset $\Omega_\rho\subset\Omega_\beta$ is a positive invariant set. The implications of this assertion is that the evolution of~\eqref{eq:overall_diff_inclusion} is contained within one of this level sets. This is a nonlinear generalisation to the one presented in \cite{BaldiviesoMonasterios2018}.
\begin{theorem}
  Suppose Assumptions~\ref{assum:preview_information}-- \ref{assump:pd}, \ref{assum:bounded_disturbance} 
\label{thm:positive_invariance} hold and for all $i\in\mc{V}_I$ $(z_i(0)\mb{d}_i(0))\in\Omega_{i,\beta}\subset\Gamma_i^N$ for some $\beta \geq \alpha$. The set $\Omega_{i,\beta}$ is positively invariant for the composite system~\eqref{eq:overall_diff_inclusion}.  
\end{theorem}
\begin{proof}
  Consider $(z_i,\mb{d}_i)\in\Omega_{i,\beta}$, from Theorem~\ref{thm:feasibility_tail}, the optimisation problem remains feasible when the available exogenous sequence assumes the $\tilde{k}^\ts{th}-$tail. Following classical results from the MPC literature, feasibility of the optimisation problem implies stability. A consequence of this is that the value function is a Lyapunov function, \ie $\nu_i^{N,0}(z_i^+,\tilde{d}_{\tilde{k}}(\mb{d}_i)) \leq \nu_i^{N,0}(z_i,\mb{d}_i) - \theta_{3,i}(|z_i|)$ and $\theta_{1,i}(|z_i|)\leq \nu_i^{N,0}(z_i,\mb{d}_i) \leq \theta_{2,i}(|z_i|)$ for some $\theta_{3,i},\theta_{2,i},\theta_{1,i}\in\mc{K}$. On the other hand, the continuity of $\nu_i^{N,0}$ over a compact set implies that is uniformly continuous on that set. From \cite[Lemma 1]{Limon2009}, there exists a $\mc{K}_\infty-$function $\alpha\sigma_{\nu,i}$ such that \[\nu_i^{N,0}(z_i^+,\mb{d}_i^+) \leq \nu_i^{N,0}(z_i,\mb{d}_i) - \theta_i(|z_i|) + \sigma_{\nu,i}(|\mb{d}_i^+ - \tilde{d}_{\tilde{k}}(\mb{d}_i)|)\]Using Assumption~\ref{assum:bounded_disturbance}, we obtain \[
    \begin{split}
      \nu_i^{N,0}(z_i^+,\mb{d}_i^+) &\leq  \nu_i^{N,0}(z_i,\mb{d}_i) - \theta_{3,i}(|z_i|) + (\rho_i - \gamma_i)(\alpha_i)\\
      & \leq (id - \theta_{3,i}\circ\theta_{2,i}^{-1})\nu_i^{N,0}(z_i,\mb{d}_i) + (\rho_i - \gamma_i)(\beta_i)\\
      & \leq (id - \theta_{3,i}\circ\theta_{2,i}^{-1})(\beta_i) + (\rho_i - \gamma_i)(\beta_i)
    \end{split}
  \]
  Taking $\gamma_i = id - \theta_{3,i}\circ\theta_{2,i}^{-1}$ yields $\nu_i^{N,0}(z_i^+,\mb{d}_i^+) \leq \beta_i$ which implies that $(z_i^+,\mb{d}_i^+)\in\Omega_{i,\beta}$
\end{proof}

The main assertion of this section is a consequence of theorem~\ref{thm:positive_invariance}.

\begin{corollary}[Recursive Feasibility] If $z_i(0)\in\{z_i\colon (z_i,\mb{d}_i)\in\Omega_{i,\beta}\}\subset\mc{Z}_i^N(\mb{d}_i)$, then $z_i(k)\in\mc{Z}_i^N(\mb{d}_i(k))$ provided the exogenous inputs update rate is limited. 
\end{corollary}
Recursive feasibility is obtained then as a consequence of the stability properties of the predictive controllers. One of the key components of the analysis is given by Assumption~\ref{assum:bounded_disturbance}, this assumption is not strong since the rate of change of the exogenous renewable injection can be controlled using the control inputs $\Delta_{h,i}$ of \eqref{eq:power_injection}, and the voltages depend on the Laplacian of the connectivity graph which (provided the current injections are bounded) implies neighbouring voltages remain close to the kernel $\ker \mc{L} = \ts{span}{~x{1}_{|\mc{V}_I|+1}}$. The negotiation framework modifies the cost according to the change in exogenous inputs, a variation in load demand or renewable injection will result in a potential update in the pricing policy. If Assumption~\ref{assum:bounded_disturbance} holds, then the controller can handle variations in its cost criteria. On the other hand, the sources of potentially large variations arise from sudden load demands; renewable injections can be kept within prescribed variation using the available inputs $\Delta_i$ and local storage. This section established conditions to guarantee feasibility of the \ac{OCP} with respect to changing forecasts. These conditions follow from regularity, see \ref{prop:Z_usc} and \ref{prop:value_function_continuity}, of the optimisation problem which is robust to bounded variations of forecasts as shown in Theorem~\ref{thm:positive_invariance}.




\section{Testing and evaluation environment}
\label{sec:test-eval-envir}
In this section, we illustrate the properties of the proposed combination of negotiation framework and distributed control. First we briefly introduce the method to obtain the predicted sequences $\mb{d}_i$; we explore the coupled constraint satisfaction properties of our approach in a simplified example. Lastly, we apply the proposed approach to a \ac{MG} with a meshed topology and analyse its performance according to its effects on power flows, voltages, pricing and overall costs. \vspace{-0.5cm}
%
\subsection{Forecasting}
\label{sec:forecast}
Assumption~\ref{assum:preview_information} lays the foundations of the proposed approach; the distributed MPC controller at each time $k$ has access to forecasts of generation and consumption of power along a finite horizon, \ie $\{d(k),d(k+1),\ldots,d(k+N)\}$. In this section, we briefly discuss how we construct these exogenous inputs $d_i = (\{w_{h,i}\}_{h\in\mc{H}},S_{L,i})\in\mbb{D}_i$ where $\mc{H}$ denotes a set of the available renewable source, \ie wind, \ac{PV}, etc. The requirement of $N-$step prediction sequences is in line with the seasonality of consumption and renewable injection; therefore the Seasonal \ac{ARIMA} represents a natural choice for forecasting. Initially introduced in \cite{box1976time}, \ac{ARIMA} is a well known forecasting method that uses  linear combinations of past data $d_i(k-j)$ and errors $e_i(k-j)$ for $j\in\Nset$. These linear combinations are polynomials $T(\cdot)$ and $S(\cdot)$ of degree $p$ and $q$ with respect to the backward shifting operator, \ie  $\Delta^{j}d(k) = d(k-j)$. The prediction and forecasting error relation at time $k$ is $T(\Delta) d_i(k) = S(\Delta)e_i(k)$ with $T(\Delta)$ containing factors of $(1-\Delta)^d$. The numbers $(p,d,q)$ completely characterise this method. The Seasonal \ac{ARIMA} (SARIMA) employs the operator $\Delta_D d_i(k) = d_i(k-D)$, then the forecasting model uses $\tilde{T} = T(\Delta_D)$ and $\tilde{S} = S(\Delta_D)$.

Both \ac{ARIMA} and S\ac{ARIMA} are commonly applied time-series based statistical models with low processing requirements as compared to deep neural networks models. The work in \cite{fang2016evaluation} shows that even without considering social and environmental factors S\ac{ARIMA} can achieve an Average Mean Percentage Error (AMAPE) of $9.44\%$ on load forecasting with a 72h time-frame. Due to the lack of daily or weekly seasonality in wind turbine generation ARIMA models are more suitable for forecasting, with the work in \cite{kavasseri2009day} showing a forecast accuracy for the square root of the Mean Square Error $(\sqrt{MSE})$ of $11.87\%$ with a 48h forecast window. 
For our testing data, the selected models for each component together with their Akaike information criterion (AIC) and Normalised Root Mean Square Error (NRMSE) are shown in Table \ref{tab:for_info}. The AIC value is based on training the models across all the scenarios and the NRMSE value is the average across all the scenarios when applied daily as a 24h forecast.\vspace{-0.3cm} 
\begin{table}[b!]
  \setlength{\extrarowheight}{1pt}
\caption{Training parameters and accuracy of forecasting models as applied to the testing scenarios}
\label{tab:for_info}
\begin{tabular}{p{2.11cm} c c c c}
\hline
 Parameter & \begin{tabular}[c]{@{}l@{}}Model\\ Type\end{tabular} & \begin{tabular}[c]{@{}l@{}}Parameters\\ (p,d,q)(P,D,Q)\end{tabular} & AIC & NRMSE \\ \hline
Consumption                   & SARIMA & (1,1,1)(0,1,1)  & 349.5  & 12.0\% \\
WT  Generation       & ARIMA & (2,0,1)  & 4344.9 & 23.3\% \\
PV Generation  & SARIMA & (1,1,2)(0,2,2) & 2177.1 & 12.4\% \\ \hline
\end{tabular}
\end{table}

\subsection{Voltage constraints}
\label{sec:voltage-constraints}
\begin{figure}[t!]
  \centering
  \input{figs/vol_cons.tikz}
  \caption{Normalised line current norm $||i_{e,2}||_\infty$ for a system with coupled constraints~\eqref{fig:volt:icd} and the system without line current constraints~\eqref{fig:volt:ind}. The upper bound  \eqref{fig:volt:imax} is $I_{e,2}^\ts{max} = 35\si{\ampere}$.\vspace{-0.7cm}}
  \label{fig:volt_constraints}
\end{figure}
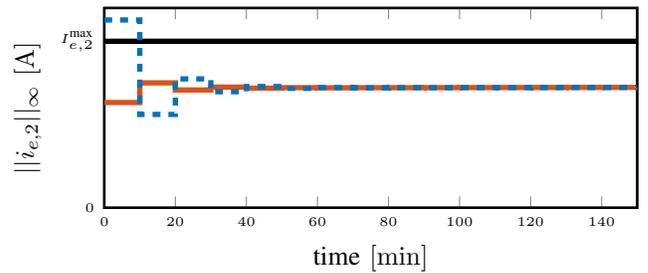
To visualise the constraint satisfaction properties , we consider initially a simplified example that comprises a network with $\mc{V} = \{0,1,2\}$ and $\mc{E}= \{(0,1),(1,2))\}$. The voltage at each node $v_i = (v_{d,i},v_{q,i})\in\Rset^2$ is constrained to a set $\mbb{V}_i = \{v_i\colon |v_i|^2 \leq 230\sqrt{2}\ \wedge v_{d,i} \geq 205\sqrt{2}\}$ satisfying Assumption~\ref{assum:network_constraints}. In addition, we consider line constraints on $e_2=(1,2)$ such that $\mbb{I}_E = \{(i_{e,1},i_{e,2})\colon |i_{e,2}|_\infty\leq 35\}$. Following the discussion of Section~\ref{sec:network-model}, these line constraints induce voltage constraints as shown in \eqref{eq:voltage_constraints_network}. The geometry of these constraints is hard to visualise since $\mbb{V}_I\subset\Rset^4$; the only available visualisations are merely projection which may lose information. To visualise the effect of these constraints, we focus on the line currents which are computed explicitly as \[i_E = Y_E(\mc{B}_{I}v_I + \mc{B}_{0}(\mc{L}_{0}+Z_0^{-1})(Z_0^{-1}E_0 - \mc{B}_0^\top Y_E\mc{B}_Iv_I))\]
The network is initialised at $v_1(0) = (315.1,-67.31)$ and $v_2=(323.6,-13.7)$ such that $(v_1(0),v_2(0))\in\mbb{V}_I$. In Figure~\ref{fig:volt_constraints}, we show the behaviour of the line currents with and without these constraints when subject to the same load demand and renewable input. The response that does not contain these coupled constraints shows more aggressive towards the a steady state value; on the other hand the proposed approach is more conservative in the transient but achieves the same result in steady state. Although a simplified network is considered here for better visualisation of the methods and results, our methodology can easily be extended to more complicated networks in a straightforward manner as seen in the next section.







\subsection{Data Sources and Scenario Generation}
\label{sec:data_sources_and_scenario}




In this section, we analyse the performance of the proposed algorithm that combines the distributed MPC controller and market layer subject to different methods of forecasting. Our study considers scenarios that cover a wide range of \ac{DER} configurations and renewable inputs. For the wind turbine data, the scenarios were generated based on the data collected in \cite{verba2020plat} from real turbines in South Wales,UK. The data for PV panels is generated following a Clear-sky models and a burr distribution noise for Sheffield,UK. The consumption data is based on a data-set from \cite{Dua_2019}. The primary parameter that is iterated through is the time of year which varies between $[Summer,Spring/Autumn,Summer]$. This influences the week selection in the PV panels and sets the week for the consumption and wind turbine as well. The size of the panels varies between $[15.4m^2, 45.6m^2, 85.2m^2]$. The wind turbine selects a value in the time of the year that matches one of the 4 identified patterns in the data. The battery sizes vary between $[5 kWh, 13.5 kWh, 25 kWh]$. The number of units in each scenario vary between $[ 3, 6 , 9 ]$; all of these nodes are connected through a physical network, see Figure~\ref{fig:adept-data}. The combination of the resulting $447$ units into $74$ scenarios with all of them running for a week enables us evaluate the influence of the various forecasting methods and their effects on power balance, voltages, pricing and overall costs. \vspace{-0.3cm}
\begin{figure}[t!]
\centering
\input{figs/DMPC_trading.tikz}
\caption{Topology of the proposed network in Section~\ref{sec:data_sources_and_scenario}. The utility grid is represented by $\Sigma_0$ and each $\Sigma_i$ is a node comprising renewable sources and loads.\vspace{-1cm}}
\label{fig:adept-data}
\end{figure}
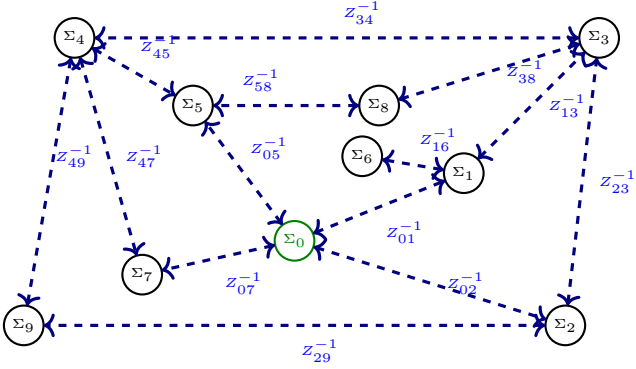

\subsection{Scenario Evaluation and Results}
\label{sec:eval_res}

The variation of the prices and cost that an individual \ac{DER} Unit received based on the type of forecast that was used can be seen in Figure~\ref{fig:for_eff_price}, where the top figures shows how the lack of forecasting will result in the unit needing to purchase energy at peak consumption times, while in the case of an existing forecast, these peaks can be predicted and energy can be bough more consistently. The lack of forecasting also influences the energy price in the scenario as unforeseen demand can increase the price of available energy.
\begin{figure}[t!]
  \centering
  \includegraphics[width=0.95\linewidth,height = 0.8\linewidth ]{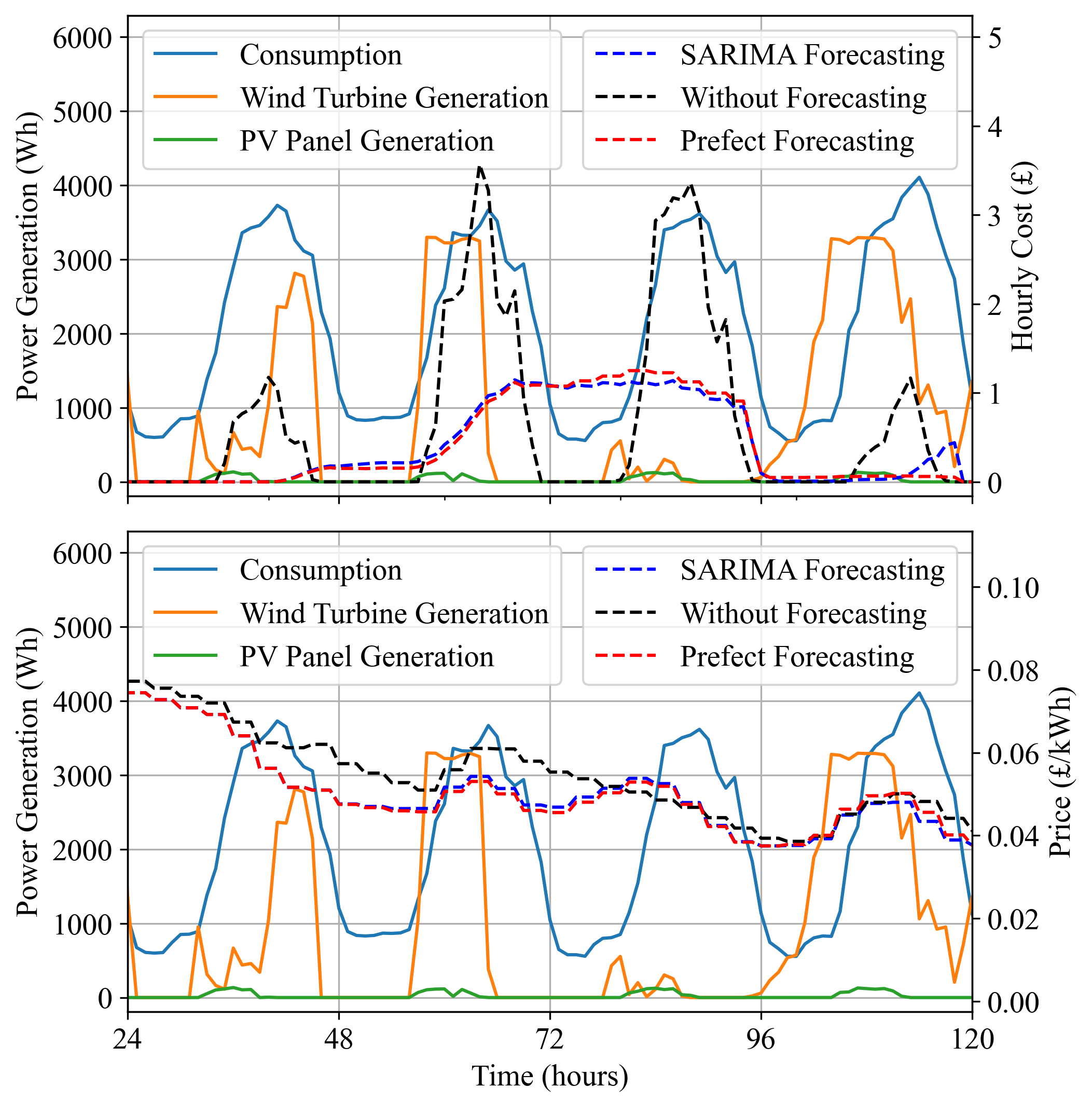}
  \caption{The effects of forecasting on the energy cost and price when considering a Single \ac{DER} given a set consumption and generation pattern. }
  \label{fig:for_eff_price}
\end{figure}
%
    
    
\begin{figure}[t!]
\centering
\includegraphics[width=0.95\linewidth]{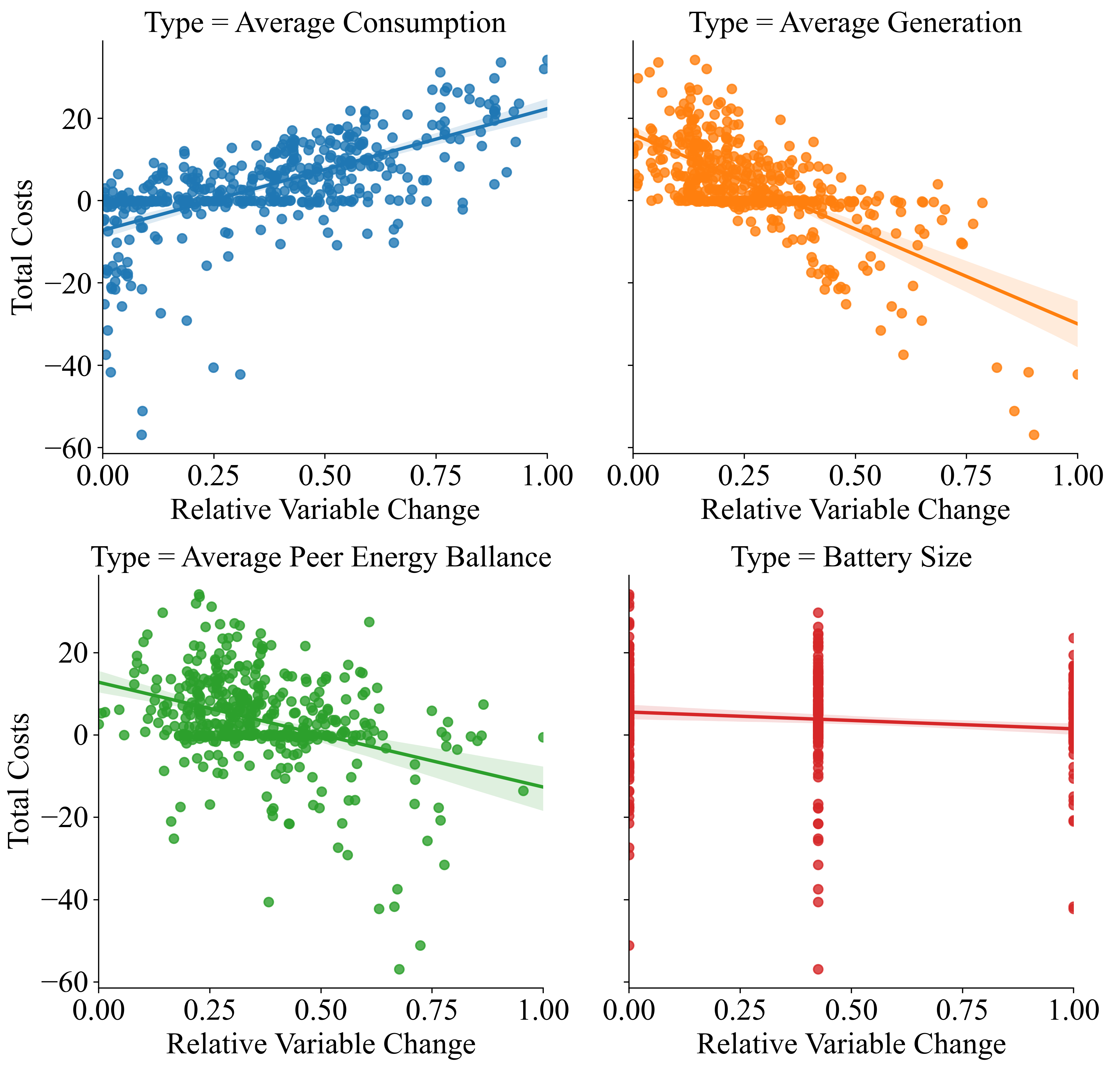}
\caption{Variation in Total Cost over the 7 day period for each individual unit based on normalised parameter values. The trends for each are highlighted using a linear regression approximation. \vspace{-1cm}}
\label{fig:overview_change}
\end{figure}

The variation of overall costs for each unit in all the scenarios as compared to the their generation, consumption, average peer energy balance and battery size can be seen in Figure~\ref{fig:overview_change}. The first two figures show the change of overall costs based on consumption and generation where we can see that while there is an obvious trend outliers show that other parameters are also at play. The effect of the peer units behaviour on price shows that as available energy in the peer network increases the total costs also decrease as purchasing energy becomes cheaper. The final plot shows that in these scenarios the battery size has an influence on the quality of results, but at these scales the $25kW$ battery provides little advantage over the $5kW$ battery. \par



\begin{table}[b!]
\centering
\caption{Overview of the benefits of SARIMA and perfect Forecasting}
\label{tab:results}
\begin{tabular}{@{}lccccc@{}}
\toprule
\multicolumn{1}{l}{\begin{tabular}[l]{@{}c@{}}Forecast \\ Type\end{tabular}} &
  \begin{tabular}[c]{@{}c@{}}Mean Cost\\ Difference\end{tabular} &
  \begin{tabular}[c]{@{}c@{}}Units \\ Worse Off\end{tabular} &
  \begin{tabular}[c]{@{}c@{}}Units \\ Better Off\end{tabular} &
  \begin{tabular}[c]{@{}c@{}}Worst \\ Unit\end{tabular} &
  \begin{tabular}[c]{@{}c@{}}Best\\ Unit\end{tabular} \\ \midrule
\begin{tabular}[l]{@{}l@{}}Naive/\\ SARIMA\end{tabular} &
  0.152 &
  170 &
  277 &
  -4.462 &
  5.152 \\
\begin{tabular}[l]{@{}l@{}}Naive/\\ Perfect\end{tabular} &
  0.207 &
  147 &
  300 &
  -3.725 &
  5.261 \\ \bottomrule
\end{tabular}
\end{table}

If we consider all the 74 scenarios, the total cost of running the units is £$1720.99$ for the naive forecasting, £$1652.69$ for the SARIMA based forecasting and £$1628.04$ for the perfect forecast. This shows that the inclusion of forecasting into control methods and markets on a similar setup has the potential to improve costs by $5.41\%$ in the case of perfect forecasting and can reduce costs by $3.96\%$ even when using simple methods such as ARIMA and SARIMAX. An overview the effects of forecasting on the whole scenario set can be seen in Table \ref{tab:results}.\par

\begin{table}[b!]
\centering
\caption{Mean Values of Unit specific quantiles for energy supply, demand and voltage}
\label{tab:sup_dem_voltage}
\begin{tabular}{@{}lcccc@{}}
\toprule
\multicolumn{1}{c}{\begin{tabular}[c]{@{}c@{}}Forecast \\ Type\end{tabular}} &
  \begin{tabular}[c]{@{}c@{}}Mean\\ Power Supply\\ .98 Quantile\end{tabular} &
  \begin{tabular}[c]{@{}c@{}}Mean\\ Power Demand\\ .98 Quantile\end{tabular} &
  \begin{tabular}[c]{@{}c@{}}Mean\\ $V_{RMS}$ Sag\\ .98 Quantile\end{tabular} &
  \begin{tabular}[c]{@{}c@{}}Mean\\ $V_{RMS}$\end{tabular} \\ \midrule
Naive   & 2818.22 & 1425.09 & 214.38 & 218.50 \\
SARIMA  & 2635.11 & 1435.49 & 214.97 & 218.62 \\
Perfect & 2562.41 & 1420.35 & 215.11 & 218.64 \\ \bottomrule
\end{tabular}
\end{table}
When considering the effects of forecasting on the stability of the resulting local grids we consider the $.98$ quantile values for power supply and demand for each unit as well as their mean voltage and $.98$ quantile of their $V_{RMS}$ Sag. We use the quantile values instead of the minimum values to filter out peaks in the system and to better understand the overall behaviour. The behaviour of each forecasting method can be seen in Table~ \ref{tab:sup_dem_voltage}, where we can see that while the effect of the forecasting methods is minimal on the Power Demand, its effects are considerable on the Supply and $V_{RMS}$ Sag.\vspace{-0.3cm}

\section{Conclusions and future work}
\label{sec:concl-future-work}

We proposed a distributed predictive controller capable of handling coupled constraints which optimises generation costs. These costs are obtained via a negotiation framework based on a \ac{MAS}; agents corresponding to each network agree upon pricing policies to transact their output power. Both control and market layers are subject to exogenous inputs dictating the interaction of the system with its environment. A rigorous analysis of the properties, existence of game theoretic equilibrium points for the market layer and recursive feasibility for the control layer, was given. The controller is proven to be recursively feasible in presence of time-varying information, both voltages form neighbouring nodes and forecasts. We have developed a testing and evaluation environment to assess the performance of our controller. We explored the effects of varying forecast accuracy in the controller performance and the pricing policies. The proposed system is a scalable solution to the problem of pricing and control. An avenue of future research is to investigate how global power flow phenomena, such as the existence of loop flows, affects our setting.\vspace{-0.3cm}

\section{Acknowledgements}
This work was supported by the UK Research and Innovation (UKRI) through the Engineering and Physical Sciences Research Council (EPSRC) as part of the Energy Revolution Research Consortium (ERRC) with the reference EP/S031863/1. \vspace{-0.8cm}

\bibliography{extracted.bib}
\begin{IEEEbiography}[{\includegraphics[width=1in,height=1.25in,clip,keepaspectratio]{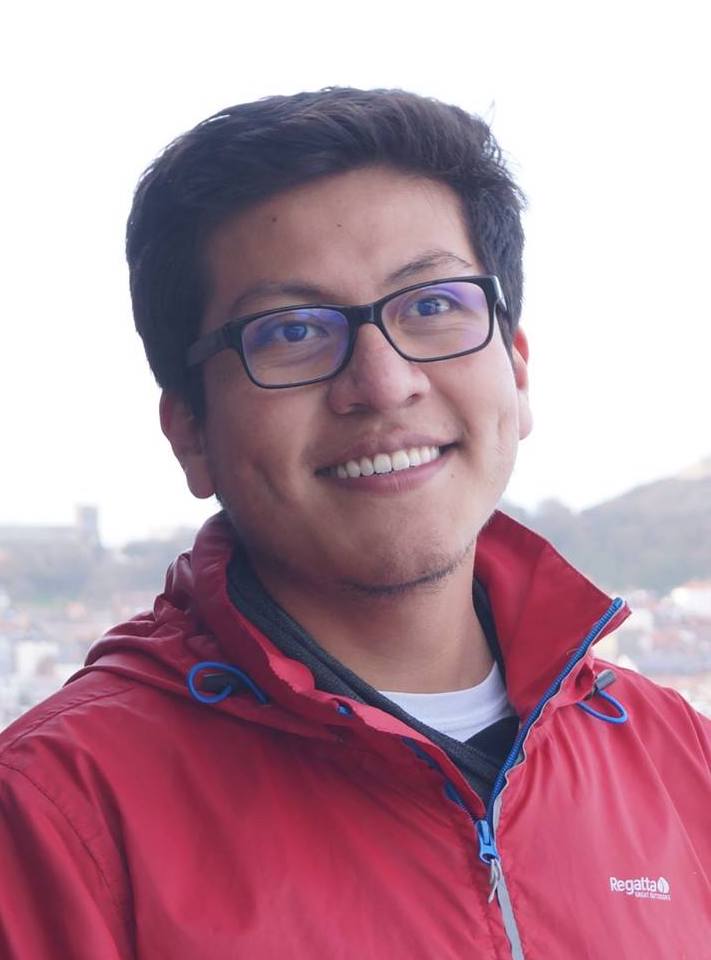}}]{Pablo R. Baldivieso-Monasterios} is a post-doctoral research associate in the Department of Automatic Control and Systems Engineering, University of Sheffield, UK. He received a Ph.D in robust distributed model predictive control from the University of Sheffield, UK in 2018. He is part of the Energy Revolution consortium (EnergyREV). His research interests include robust and distributed model predictive and optimisation-based control, and game theoretic methods for control and smartgrids.
\end{IEEEbiography}
\begin{IEEEbiography}[{\includegraphics[width=1in,height=1.25in,clip,keepaspectratio]{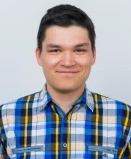}}]{Nandor Verba} is a research fellow at Coventry University. He received a PhD in the field of Internet of Things from Coventry University, UK for his work on Application Deployment Optimisation in heterogeneous Fog Computing environments. His current work on HEED and EnergyREV focuses on the digital architecture of Energy Systems and achieving their full potential through cyber-physical components with special interest in forecasting, anomaly detection and optimisation.
\end{IEEEbiography}
\vspace{-1cm}
\begin{IEEEbiography}[{\includegraphics[width=1in,height=1.25in,clip,keepaspectratio]{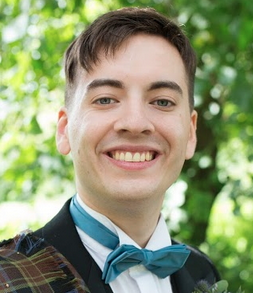}}]{Euan Morris} received the MEng degree (2014) and the PhD degree (2020) from the Department of Electronic and Electrical Engineering at the University of Strathclyde. He is currently a Research Associate at the both the Power Networks Demonstration Centre and with the EnergyREV research consortium with a focus on advanced artificial intelligence applications.
\end{IEEEbiography}
\vspace{-1cm}
\begin{IEEEbiography}[{\includegraphics[width=1in,height=1.25in,clip,keepaspectratio]{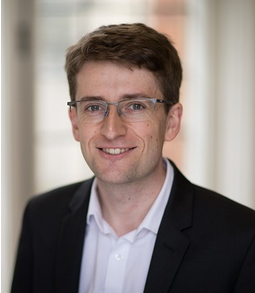}}]{Thomas Morstyn}  (M’16) received the B.Eng. degree (Hons.) in electrical engineering from the University of Melbourne in 2011, and the Ph.D. degree in electrical engineering from the University of New South Wales in 2016. He is a Lecturer of Power Electronics and Smart Grids with the School of Engineering, University of Edinburgh. He is also an Oxford Martin Associate with the Oxford Martin School, University of Oxford. His research interests include multiagent control and market design for integrating distributed energy resources into power system operations.
\end{IEEEbiography}
\vspace{-1cm}
\begin{IEEEbiography}[{\includegraphics[width=1in,height=1.25in,clip,keepaspectratio]{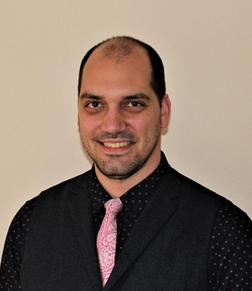}}]{George Konstantopoulos} received his Diploma and Ph.D.in Electrical and Computer Engineering from the University of Patras, Greece, in 2008 and 2012, respectively.Since 2013, he has been with the Department of Automatic Control and Systems Engineering, The University of Sheffield,UK, where he is currently a senior Lecturer. His current research interests include nonlinear modelling, control and stability analysis of power converter and electric machine systems with emphasis in microgrid operation, renewable energy systems and motor drives.
\end{IEEEbiography}
\vspace{-1cm}
\begin{IEEEbiography}[{\includegraphics[width=1in,height=1.25in,clip,keepaspectratio]{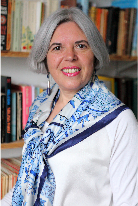}}]{Elena Gaura} is a Professor of Pervasive Computing at Coventry University, UK, researching into Cyber Physical Systems. She works on evidence driven designs for socio-technical complex systems governed by data and specifically through data-to-knowledge pipelines. Over the past few years she lead numerous international programmes supporting the development of science and innovation as well as developing capacity for research world wide and specifically in and with the Global South.
\end{IEEEbiography}
\vspace{-1cm}
\begin{IEEEbiography}[{\includegraphics[width=1in,height=1.25in,clip,keepaspectratio]{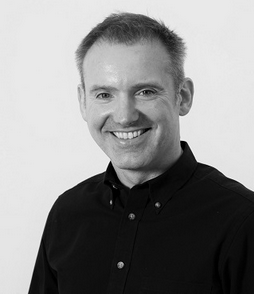}}]{Stephen McArthur} is the Distinguished Professor of Intelligent Energy Systems and a Deputy Associate Principal (for Research, Knowledge Exchange and Innovation) at the University of Strathclyde. His main area of expertise is intelligent system applications in energy covering smart grids, condition monitoring and data analytics. Data analytic solutions developed by his team have been deployed for a range of national and international energy companies. As Principal Investigator of the EnergyREV consortium, Stephen is responsible for ensuring that the programme aligns with the needs of PFER and delivers on its objectives and KPIs. His research within EnergyREV targets the use of artificial intelligence to support the smart functionality and interoperability required within smart local energy systems. \end{IEEEbiography}
\end{document}

%% file: figs/local_diagram.tikz
\ctikzset{bipoles/resistor/height=0.15}
\ctikzset{bipoles/resistor/width=0.4}
\ctikzset{bipoles/capacitor/height=0.4}
\ctikzset{bipoles/capacitor/width=0.15}
\tikzset{%
wind turbine/.pic={
  \tikzset{path/.style={fill = black!80!white, line join=round,scale = 0.4}}
    \path [path] 
    (-.25,-1.5) arc (180:360:.25 and .0625) -- (.0625,1) -- (-.0625,1) -- cycle;
  \foreach \i in {90, 210, 330}{
      \path [path, shift=(90:1), rotate=\i]
        (0,-0.125) arc (180:0:0.65 and 0.125) -- ++(0,0.125) arc (0:180:0.65 and 0.25) -- cycle;
  }
  \path [path] (0,1) circle [radius=.15];
}}

\tikzset{%
pv panel/.pic={
  \tikzset{path/.style={fill = black!80!white, draw=white, line join=round, scale = 0.3}}
  	\def\rot{60};
    \path [path] 
    (-.25,-2.5) arc (180:360:.25 and .0625) -- (.0625,0) -- (-.0625,0) -- cycle;
  \path [path] 
    ({-2.5-tan(\rot)},1) -- ({2.5-tan(\rot)},1) -- ({2.5+tan(\rot)},-1) -- ({-2.5+tan(\rot)},-1) -- cycle;
    
    \foreach \i in {1,2,3,4,5}{
    \foreach \j in {1,2,3}{
    \path [path] ({(-2.5+\i*0.8)-(-1+\j*0.5)*tan(\rot)},{-1+\j*0.5}) circle [fill = red, radius=.09];
    
	\path [path] 
    ({-2.5-(-1+\j*0.5)*tan(\rot)},-1+\j*0.5) -- ({2.5-(-1+\j*0.5)*tan(\rot)},-1+\j*0.5);
    }
  \path [path] 
    ({(-2.5+\i*0.8)-tan(\rot)},1) -- ({(-2.5+\i*0.8)+tan(\rot)},-1);  
    } 
}}

\tikzset{%
battery/.pic={
  \tikzset{path/.style={fill = black!20!white, draw=black, ultra thick, line join=round, scale = 0.3}}
  \path [path] 
    (-2.5,1.5) -- (-1.5,1.5) -- (-1.5,2)-- (-0.5,2) --
    (-0.5,1.5) -- (0.5,1.5) -- (0.5,2) -- (1.5,2) -- 
    (1.5,1.5)  -- (2.5,1.5) --
    (2.5,-1.5) -- (-2.5,-1.5) -- cycle;
    
    \path [path] 
    (-1,1.3) -- (-1,0.7);
	\path [path]
    (-1.3,1) -- (-0.7,1);
    
	\path [path]
    (1.3,1) -- (0.7,1);
}}

\tikzset{%
load/.pic={
  \tikzset{path/.style={fill = black!80!white, draw=white, line join=round, scale = 0.3}}
  \path [path] 
     (-2,1.5) -- (-0.5,3)-- (-0.5,1.5) --
	 (-0.5,1.5) -- (1,3)-- (1,1.5) --
	 (1,1.5) -- (2.5,3)-- (2.5,1.5) --
	 (2.5,-1.5) -- (-1,-1.5) -- (-1,-0) -- (-2,0) -- (-2,-1.5) -- 
	 (-2.5,-1.5) -- (-2.5,1.5) -- cycle;
}}

\tikzset{%
inverter/.pic={
  \tikzset{path/.style={fill = white, line join=round, scale = 0.4, draw = black, thick}}
    \path [path] 
    (1.3,1.8)  -- (-1.3,1.8) -- (-1.3,-1.8) -- (1.3,-1.8) -- cycle;
    \path [path] 
    (-1.3,1.8)  -- (1.3,-1.8);
    \ifcase#1
    \or
    \path [path]  (-1,-1) sin (-0.75,-0.75); 
	\path [path]  (-0.75,-0.75) cos (-0.5,-1);
	\path [path]  (-0.5,-1) sin (-0.25,-1.25);
	\path [path] (-0.25,-1.25) cos (0,-1);
    
	\path [path] (0,1) sin (0.25,1.25); 
	\path [path] (0.25,1.25) cos (0.5,1);
	\path [path] (0.5,1) sin (0.75,0.75);
	\path [path] (0.75,0.75) cos (1,1);
    
    \or
	\path [path] (-1,-1) sin (-0.75,-0.75); 
	\path [path] (-0.75,-0.75) cos (-0.5,-1);
	\path [path] (-0.5,-1) sin (-0.25,-1.25);
	\path [path] (-0.25,-1.25) cos (0,-1);
    
    \path [path] (0,1.25)  -- (1,1.25);
    \path [path, dashed] (0,1)  -- (1,1);
    
    \or
    \path [path] (-1,-1)  -- (0,-1);
    \path [path, dashed] (-1,-1.25)  -- (0,-1.25);
    
	\path [path] (0,1) sin (0.25,1.25); 
	\path [path] (0.25,1.25) cos (0.5,1);
	\path [path] (0.5,1) sin (0.75,0.75);
	\path [path] (0.75,0.75) cos (1,1);
    
    \or
    \path [path] (0,1.25)  -- (1,1.25);
    \path [path, dashed] (0,1)  -- (1,1);
    
    \path [path] (-1,-1)  -- (0,-1);
    \path [path, dashed] (-1,-1.25)  -- (0,-1.25);
    
	\fi
}}


\begin{circuitikz}[scale=0.6,font=\scriptsize]

\draw 
(3,-3) node{$\ts{Load}$}

(1,1.75) to (2,1.75) to[semithick,/tikz/circuitikz/bipoles/length=1cm,sV] (4,1.75)
(3,3.2) node[anchor=north]{$PCC$}

(-5,3) to[/tikz/circuitikz/bipoles/length=1cm,L,i =$P_\ts{pv}$] (-1.5,3)

(-5,0) to[/tikz/circuitikz/bipoles/length=1cm,L,i =$P_\ts{bat}$] (-1.5,0)

(-5,-3) to[/tikz/circuitikz/bipoles/length=1cm,L,i =$P_\ts{wind}$] (-1.5,-3)

(-1.5,1.75) to[/tikz/circuitikz/bipoles/length=1cm,L] (1,1.75)
(0,2.7) node[anchor=north]{$P_\ts{grid}$}

(-1.5,3) to (-1.5,-1.75)
(-1.5,-3) to (-1.5,-1.75) to (1,-1.75)
;

\draw[ultra thick] (-8,3) -- (-5.5,3);
\path (-8,3) pic {pv panel};
\path (-5,3) pic {inverter = 3};

\draw[ultra thick] (-8,0) -- (-5.5,0);
\path (-8,0) pic {battery};
\path (-5,0) pic {inverter = 3};

\draw[ultra thick] (-8,-3) -- (-5.5,-3);
\path (-8,-3) pic {wind turbine};
\path (-5,-3) pic {inverter = 1};

\path (3,-2) pic {load};
\draw[ultra thick] (1,-0.75) -- (1,-2.75);
\draw[ultra thick] (1,-1.75) -- (2.2,-1.75);

\draw[ultra thick] (1,2.75) -- (1,0.75);
\draw[ultra thick] (4,2.25) -- (4,1.25);


\end{circuitikz}

%% file: figs/global_diagram.tikz
\ctikzset{bipoles/resistor/height=0.15}
\ctikzset{bipoles/resistor/width=0.4}
\ctikzset{bipoles/capacitor/height=0.4}
\ctikzset{bipoles/capacitor/width=0.15}

\tikzset{%
adept/.pic={
  \tikzset{path/.style={fill = black!10!white, draw=black, semithick, line join=round, scale = 1}}

  \path [path] 
    (-0.55,0.3) -- (0.55,0.3) -- (0.55,-0.3)-- (-0.55,-0.3) -- cycle;
	\node (0,0) {$\Sigma_{#1}$};
}}

\def\Num{10}

\begin{circuitikz}[scale=0.8]


\foreach \x [count=\xc] in {1,2,...,\Num} { 

\ifthenelse{\x = 6}
           {
\draw ({2.5*cos(((360*(\x-1))/\Num))},{2.5*sin(((360*(\x-1))/\Num))}) to ({1*cos(((360*(\x-1))/\Num))},{1*sin(((360*(\x-1))/\Num))});
	\draw ({4.5*cos(((360*(\x-1))/\Num))},{4.5*sin(((360*(\x-1))/\Num))}) to[/tikz/circuitikz/bipoles/length=1cm,L,i =$P_{grid}$] ({2.5*cos(((360*(\x-1))/\Num))},{2.5*sin(((360*(\x-1))/\Num))});
	\draw ({5*cos(((360*(\x-1))/\Num))},{5*sin(((360*(\x-1))/\Num))}) to[/tikz/circuitikz/bipoles/length=1cm,sV]  ({4.5*cos(((360*(\x-1))/\Num))},{4.5*sin(((360*(\x-1))/\Num))});

           }
           {
	\draw ({3.5*cos(((360*(\x-1))/\Num))},{3.5*sin(((360*(\x-1))/\Num))}) to[/tikz/circuitikz/bipoles/length=1cm,L,i =$P_{\xc}$] ({1.5*cos(((360*(\x-1))/\Num))},{1.5*sin(((360*(\x-1))/\Num))}) to ({1*cos(((360*(\x-1))/\Num))},{1*sin(((360*(\x-1))/\Num))});

	\path ({4*cos(((360*(\x-1))/\Num))},{4*sin(((360*(\x-1))/\Num))}) pic {adept = \xc};

           }
%
%
%
%
}

\node[cloud, cloud puffs=7, minimum width= 2.5cm, minimum height=2.5cm, align=center, fill=black!15!white, draw = black, semithick] (cloud) at (0, 0) {};

\node[anchor = south] (0,0) {$\ts{Distribution}$};
\node[anchor = north] (0,0) {$\ts{Network}$};

%

\end{circuitikz}

%% file: figs/vol_cons.tikz
%
\definecolor{mycolor1}{rgb}{0.00000,0.44700,0.74100}%
\definecolor{mycolor2}{rgb}{0.85000,0.32500,0.09800}%
\definecolor{mycolor3}{rgb}{0.92900,0.69400,0.12500}%
\definecolor{mycolor4}{rgb}{0.49400,0.18400,0.55600}%
\begin{tikzpicture}

\begin{axis}[%
width=0.8\linewidth,
height=0.3\linewidth,
at={(0\linewidth,0\linewidth)},
scale only axis,
xmin=0,
xmax=15,
xlabel = {time $[\si{\minute}]$},
ylabel = {$||i_{e,2}||_\infty$ $[\si{\ampere}]$},
xtick = {0,2,4,6,8,10,12,14},
ytick = {0,1},
xticklabels = {$0$,$20$,$40$,$60$,$80$,$100$,$120$,$140$},
yticklabels = {$0$,$I_{e,2}^\ts{max}$},
ymin=0,
ymax=1.2,
legend style={legend cell align=left,align=left,draw=white!15!black}
]

\addplot[color=black,line width=2.0pt] plot table[row sep=crcr] {%
0       1\\
15	1\\
};\label{fig:volt:imax}

\addplot[const plot,color=mycolor2,solid,line width=2.0pt] plot table[row sep=crcr] {%
0	0.632015700365307\\
1	0.750127717962394\\
2	0.707423817167265\\
3	0.724509451413359\\
4	0.718511492612867\\
5	0.721132234889911\\
6	0.720417287910443\\
7	0.721106347855125\\
8	0.721140536787486\\
9	0.7214187889216\\
10	0.721627823183422\\
11	0.721872489040015\\
12	0.722128859636612\\
13	0.72256419336431\\
14	0.722792412856301\\
15	0.722874032739443\\
};\label{fig:volt:icd}

\addplot[const plot,color=mycolor1,dashed,line width=2.0pt] plot table[row sep=crcr] {%
0	1.12778271274036\\
1	0.560371077062865\\
2	0.774047304478978\\
3	0.696889455286738\\
4	0.72802891395019\\
5	0.717132576729065\\
6	0.72179153576074\\
7	0.720527598535845\\
8	0.721339292277261\\
9	0.721335050095978\\
10	0.721656576525809\\
11	0.721860372144398\\
12	0.722133019665973\\
13	0.722562439853866\\
14	0.722793014783041\\
15	0.722873778970551\\
};\label{fig:volt:ind}

\end{axis}
\end{tikzpicture}%

%% file: figs/DMPC_trading.tikz
\begin{tikzpicture}[scale=0.45,font=\tiny]

  \begin{scope}[anchor=west,every node/.style={circle,thick,draw,minimum size=15,inner sep=0pt, outer sep=0pt},black]
    \node (S2) at (5,2) {$\Sigma_1$};
    \node (S3) at (8,-2.5) {$\Sigma_2$};    
    \node (S4) at (9,6) {$\Sigma_3$};
    \node (S5) at (-6.5,6) {$\Sigma_4$};
    \node (S6) at (-3,4) {$\Sigma_{5}$};
    \node (S7) at (2,2.5) {$\Sigma_6$};
    \node (S8) at (-4.5,-1) {$\Sigma_7$};
    \node (S9) at (2.5,4) {$\Sigma_8$};
    \node (S0) at (-8,-2.5) {$\Sigma_9$};
  \end{scope}

  \begin{scope}[anchor=west,every node/.style={circle,thick,draw,minimum size=15,inner sep=0pt, outer sep=0pt},green!50!black]
    \node (S1) at (0,0) {$\Sigma_0$};
  \end{scope}

\begin{scope}[blue,decoration={ markings,
    mark=at position 0.5 with {\arrow{>}}},every node/.style={circle}, every edge/.style={draw=blue!50!black,very thick,dashed}]
  \path[<->] (S1) edge node[below right] {$Z_{01}^{-1}$} (S2);
  \path[<->] (S1) edge node[right] {$Z_{02}^{-1}$} (S3); 
  \path[<->] (S1) edge node[above right] {$Z_{05}^{-1}$} (S6);
  \path[<->] (S1) edge node[below right] {$Z_{07}^{-1}$} (S8);            
  \path[<->] (S2) edge node[above right] {$Z_{16}^{-1}$} (S7); 
  \path [<->] (S2) edge node[right] {$Z_{13}^{-1}$} (S4);
  \path [<->] (S3) edge node[right] {$Z_{23}^{-1}$} (S4);
  \path [<->] (S3) edge node[below right] {$Z_{29}^{-1}$} (S0);
  \path [<->] (S4) edge node[above right] {$Z_{34}^{-1}$} (S5);
  \path [<->] (S5) edge node[above right] {$Z_{45}^{-1}$} (S6);
  \path [<->] (S5) edge node[right] {$Z_{47}^{-1}$} (S8);
  \path [<->] (S5) edge node[above right] {$Z_{49}^{-1}$} (S0);    
  \path [<->] (S9) edge node[right] {$Z_{38}^{-1}$} (S4);
  \path [<->] (S9) edge node[above left] {$Z_{58}^{-1}$} (S6);
\end{scope}

\end{tikzpicture}